\documentclass[final,twocolumn]{IEEEtran}

\usepackage{graphicx}
\usepackage{amssymb}
 
\usepackage{amsmath}
\usepackage{dsfont}
\usepackage{bm}
\usepackage{url}
\usepackage{amsthm}
\usepackage{subfigure}
\usepackage{cite}

\def\R{{\mathds R}}
\def\C{{\mathds C}}

\def\e{\mathrm{e}}

\newcommand{\be}{\begin{equation}}
\newcommand{\ee}{\end{equation}}

\newcommand{\bI}{{\mbox{\boldmath $I$}}}
\newcommand{\bz}{{\mbox{\boldmath $z$}}}
\newcommand{\bn}{{\mbox{\boldmath $n$}}}
\newcommand{\bv}{{\mbox{\boldmath $v$}}}
\newcommand{\bp}{{\mbox{\boldmath $p$}}}

\newcommand{\bC}{{\mbox{\boldmath $C$}}}

\newcommand{\bR}{{\mbox{\boldmath $R$}}}

\newcommand{\bS}{{\mbox{\boldmath $S$}}}

\newtheorem{proposition}{Proposition}

\newtheorem{corollary}{Corollary}

 \newcommand{\test}{\mbox{$
 \begin{array}{c}
 \stackrel{ \stackrel{\textstyle H_1}{\textstyle >} }{ 
 \stackrel{\textstyle <}{ \textstyle  H_0} }
\end{array}
$}}

\usepackage{color}
\newcommand{\ed}{\color{black}} 

\usepackage[linesnumbered,ruled]{algorithm2e}

\SetKwInput{KwInput}{Input}                
\SetKwInput{KwOutput}{Output}              
\SetKwComment{Comment}{$\triangleright$\ }{}

\begin{document}

\title{
\ed Design of Customized Adaptive Radar Detectors in the CFAR Feature Plane}

\author{Angelo Coluccia${}^*$, \IEEEmembership{Senior Member, IEEE}, Alessio Fascista, \IEEEmembership{Member, IEEE}, and Giuseppe Ricci, \IEEEmembership{Senior Member, IEEE}
\thanks{All authors are with the Dipartimento di Ingegneria dell'Innovazione,
        Universit\`a del Salento, Via Monteroni, 73100 Lecce, Italy.
        E-Mail: name.surname@unisalento.it
        
        ${}^*$ Corresponding author: angelo.coluccia@unisalento.it}
}

\maketitle

\begin{abstract}
The paper addresses the design of adaptive radar detectors having desired behavior, in Gaussian disturbance with unknown statistics. Specifically, given detection probability specifications for chosen signal-to-noise ratios and steering vector mismatch levels, a methodology for the optimal design of {\ed customized} CFAR detectors is devised in a suitable feature plane based on maximal invariant statistics. To overcome the analytical and numerical intractability of the resulting optimization problem, a novel general reduced-complexity algorithm is developed, which is shown to be effective in providing a close approximation of the desired detector. The proposed approach solves the open problem of ensuring a prefixed false alarm probability while controlling the behavior under both matched and mismatched conditions, so enabling the design of fully {\ed customized} adaptive CFAR detectors.
\end{abstract}

\begin{IEEEkeywords}
Radar, GLRT, CFAR property, robust detectors, selective detectors, mismatched signals, feature space
\end{IEEEkeywords}

\section{Introduction}\label{sec:intro}

The detection of targets embedded in unknown Gaussian disturbance composed of thermal noise, clutter, and possible jamming interferers is a central problem in the radar detection literature.
{\ed A consolidated approach is to resort to the generalized likelihood ratio test (GLRT) approach, in which the statistics of the disturbance are estimated through the aid of a set of secondary data, having the same statistics as $H_0$.
Since the pioneering work by Brennan and Reed \cite{Brennan}, and then Kelly \cite{Kelly}, particular focus has been put on obtaining statistics that do not depend upon unknown clutter or noise statistics under the $H_0$ hypothesis: this in fact guarantees that the detection threshold can be set to ensure a prefixed false alarm probability ($P_{fa}$), a property referred to as constant false alarm rate (CFAR). }

Several detectors have been derived in the past decades based on such a rationale. Significant attention has been paid to the design of CFAR detectors with desired properties in terms of probability of detection ($P_d$) under mismatched conditions. Typically, a robust detector is desirable to cope with possible off-grid conditions due to angle and/or Doppler quantization, which imply that the actual steering vector may be not aligned with the nominal one; conversely, a selective detector is desirable to reject unwanted signals due to jamming or spectrum co-existence \cite{TSP_con_Richmond}.  In this respect,  Kelly's detector is a moderately selective receiver \cite{Kelly89}, while the adaptive matched filter (AMF) \cite{Kelly-Nitzberg} is a robust receiver. Other well-known examples of selective receivers are the adaptive coherence estimator (ACE) \cite{ACE} (also called adaptive normalized matched filter), the ABORT or whitened-ABORT (WABORT) detectors  \cite{Pulsone-Rader,Fabrizio-Farina,W-ABORT}, as well as the Rao detector \cite{RAO}.
A further type of receivers is based on the idea of inserting a parameter in a well-known statistic, so as to obtain a tunable detector: for instance, in Kalson's detector \cite{Kalson} a nonnegative parameter is introduced in the Kelly's statistic, in order to control the rejection level of mismatched signals, in between the AMF and Kelly's detector.

The design of detectors with suitable symmetries that can also ensure the CFAR behavior has found an important theoretical tool in the principle of invariance \cite{Bose_Steinhardt,Ramirez-Santamaria,Conte_invariant,DeMaio_invariant}. 
In our previous work \cite{TSP_CFAR-FP}, a ``CFAR feature plane'' (CFAR-FP) is introduced for a suitable maximal invariant of the classical adaptive radar detection problem after Kelly's formulation. {\ed In the CFAR-FP, radar returns are mapped to two-dimensional clusters whose properties in terms of position and shape in the plane can be analytically characterized and expressed as a function of few main parameters, so shedding new light on the behavior of several well-known CFAR detectors.}

One of the challenges in adaptive detection is how to ensure a prefixed $P_{fa}$ and, at the same time, control the behavior under both matched and mismatched conditions, while guaranteeing the CFAR property (as better discussed in Sec. \ref{sec:background}). Moreover, enhancing either the robustness or the selectivity of a radar detector often comes at the price of a $P_d$ loss under matched conditions \cite{BOR-MorganClaypool}. Although \cite{TSP_CFAR-FP} provided tools for  interpreting the performance of CFAR detectors, the design therein was mostly heuristic. Indeed,
the development of a general methodology for the design of fully customized adaptive CFAR detectors is still an open research problem.
 
Aiming at advancing the literature towards this direction, in the present paper an original design methodology is devised based on the CFAR-FP.
{\ed Specifically, after introducing the required background and discussing in detail the importance of having an analytical framework to guide the design of customized detectors in the CFAR-FP in Sec. \ref{sec:background}}, the following  contributions are provided:
\begin{itemize}
    \item Given a desired detection behavior, corresponding to $P_d$ specifications for chosen signal-to-noise ratios and steering vector mismatch levels (examples are discussed in Sec.~\ref{sec:examples}), the optimal infinite-dimensional design problem of customized CFAR detectors working at a preassigned $P_{fa}$ is formulated, with suitable cost function and constraint (Sec. \ref{sec:optimal_infinite}). 
    To overcome its mathematical intractability, a finite-dimensional version is considered (Sec. \ref{sec:optimal_finite}), using a suitable parametric family of approximation functions; the resulting optimization problem is however  highly-nonlinear, and state-of-the-art numerical solvers fail to provide a feasible solution (Sec. \ref{sec:choice}).
    \item  To address such a challenge, a novel general reduced-complexity algorithm is developed, {\ed assuming a piecewise-linear structure for the approximation functions that allows to characterize the involved statistics}. Specifically,  convenient analytical formulas are derived to ease the evaluation of the cost function and provide a closed-form expression for the $P_{fa}$ constraint (Sec. \ref{sec:algorithm}). Based on such results, a sub-optimal search strategy in a reduced, though sufficiently rich, solution space is devised.
      \item The effectiveness of the proposed approach is demonstrated by designing two novel detectors representative of either robust or selective behaviors (among the examples of Sec.~\ref{sec:examples}). The approximation accuracy is evaluated in comparison with a plain solution, and a performance assessment against  well-known detectors is also performed (Sec. \ref{sec::NumResults}). {\ed A thorough performance analysis conducted on both simulated and real radar data shows that the proposed approach is effective in providing a close approximation of the desired detector}, while ensuring a prefixed $P_{fa}$ and controlling the behavior under both matched and mismatched conditions. 
\end{itemize}
We conclude the paper in Sec. \ref{sec:conclusions}.

\section{Background and Motivation}\label{sec:background}

In this section we recall the classical formulation of the radar detection problem, and review its interpretation in the CFAR-FP {\ed introduced} in \cite{TSP_CFAR-FP}. This will serve to provide the necessary background, for self-consistency, and also to illustrate in more details the motivation of the present work.

\subsection{\ed Problem Formulation}

The classical hypothesis testing problem for 
detecting the possible presence of a (point-like) {\ed coherent} target from a given cell under test (CUT) is given by
\begin{equation}
\left\{
\begin{array}{ll}
H_{0}: & \bz =  \bn \\
H_{1}: & \bz = \alpha \bv + \bn
\end{array} 
\right. \label{eq:binary_test}
\end{equation}
where
$\bz \in \C^{N \times 1}$, $\bn \in \C^{N \times 1}$, and $\bv \in \C^{N \times 1}$ are 
the received vector, the overall disturbance term, and 
the known space-time steering vector of the target. The unknown deterministic parameter $\alpha \in \C$ is the target amplitude, depending on radar cross-section, multipath, and other channel effects.

Kelly \cite{Kelly} derived a GLRT for problem \eqref{eq:binary_test} assuming complex {\ed Normal} distributed $\bn$ with zero mean and unknown (Hermitian) positive definite covariance matrix $\bC$, and $K \geq N$ independent and identically distributed training (or secondary) data $\bz_1, \ldots, \bz_{K}$ (independent of $\bz$, free of target echoes, and  sharing with the CUT the statistical characteristics of the noise). Let $\bS=\sum_{k=1}^K \bz_k\bz_k^\dag$, then
Kelly's statistic is
\begin{align}
t_{\text{\tiny{Kelly}}}  &= 
 \frac{|\bz^{\dagger} \bS^{-1} \bv |^2}{\bv^{\dagger} \bS^{-1} \bv \, (1+ \bz^{\dagger} \bS^{-1} \bz )}  \label{eq:Kelly}
\end{align}
with $(\cdot)^\dagger$ the Hermitian operator and $|\cdot|$ the modulus of a complex number. Eq.~\eqref{eq:Kelly} can be rewritten as $t_{\text{\tiny{Kelly}}} = \frac{\tilde{t}}{1+\tilde{t}}$, where
\begin{equation}
\tilde{t} = \frac{|\bz^{\dagger} \bS^{-1} \bv |^2}{\bv^{\dagger} \bS^{-1} \bv \, \left(1+ \bz^\dag \bS^{-1} \bz - \frac{|\bz^{\dagger} \bS^{-1} \bv |^2}{\bz^{\dagger} \bS^{-1} \bz} \right)}\label{eq:t_tilde}
\end{equation}
hence $t_{\text{\tiny{Kelly}}}$ and $\tilde{t}$ are equivalent statistics
(and of course $ \tilde{t} = \frac{t_\text{\tiny Kelly}}{1-t_\text{\tiny Kelly}} $). As mentioned, a remarkable property of Kelly's detector is that it has the CFAR property; moreover, it behaves as a moderately selective detector when the actual steering vector $\bp$ in the received signal $\bz$ is not aligned with the nominal one $\bv$. The mismatch level is quantified by the squared cosine of the angle between these two vectors, i.e.,
\begin{equation}
\cos^2 \theta = \frac{\bp^\dag\bC^{-1} \bv}{\bp^\dag\bC^{-1}\bp \bv^\dag\bC^{-1}\bv}. \label{eq:cos2theta}
\end{equation}

Invariance theory has shown that CFAR detectors in Gaussian disturbance  can be written in terms of equivalent pairs of  maximal invariant statistics; a convenient choice adopted in \cite{TSP_CFAR-FP} is $(\beta, \tilde{t})$, where
\begin{equation}
\beta= \frac{1}{ 1+\bz^{\dagger} \bS^{-1} \bz  - \frac{|\bz^{\dagger} \bS^{-1} \bv|^{2}}{\bv^{\dagger} \bS^{-1} \bv}}.
\end{equation}
Therefore, the test associated to a generic CFAR detector $X$ (including AMF, ACE, Kalson, etc.) can be rewritten as
\begin{equation}
t_X(\beta, \tilde{t}) \test \eta_X \label{eq:test1}
\end{equation}
where  $t_X$ is the decision statistic and $\eta_X$ is the threshold that guarantees the desired $P_{fa}$.

\subsection{\ed Detection in the CFAR Feature Plane}

It has been shown in \cite{TSP_CFAR-FP} that in most cases the curve $t_X(\beta, \tilde{t}) = \eta_X$ in the $\beta$-$\tilde{t}$ CFAR-FP can be made explicit, meaning that \eqref{eq:test1} is equivalent to
\begin{equation}
\tilde{t}  \test f_X(\beta; \eta_X) \label{eq:test2}
\end{equation}
where $f_X$ is called \emph{decision region boundary} and separates the plane in two regions: for data points falling in the bottom-most part the  detector will decide for  $H_0$, while for the upper-most part it will decide for $H_1$, as shown in Fig. \ref{fig:FP}.

\begin{figure}
\centering
\includegraphics[width=8cm]{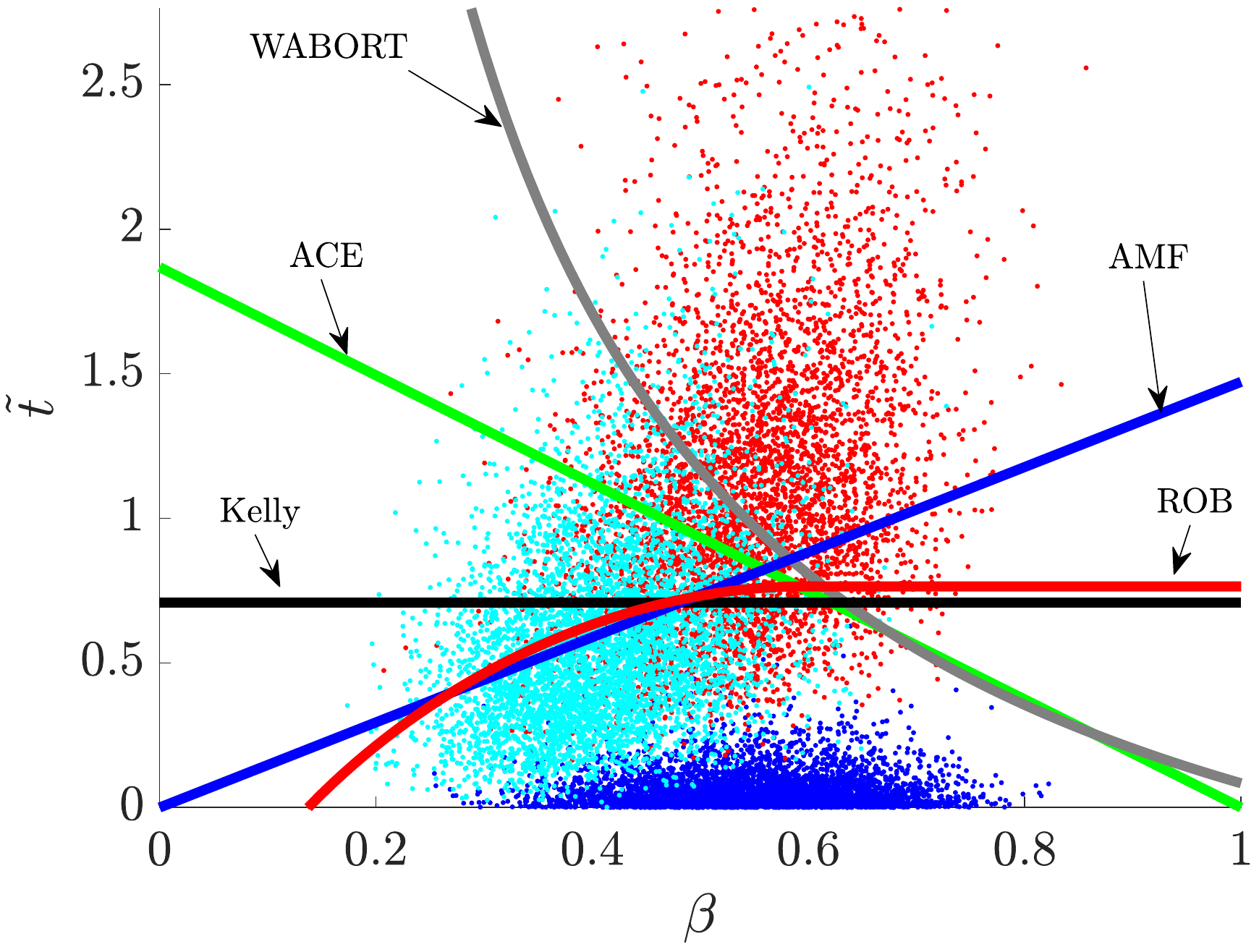}
\caption{Some well-known detectors in the CFAR-FP, for $N=8$, $K=16$, $P_{fa}=10^{-4}$. Blue dots are the $H_0$ cluster, red dots are the $H_1$ cluster for SNR = 15 dB, while cyan dots represent a  cluster under mismatched conditions ($\cos^2 \theta=0.65$).}\label{fig:FP}
\end{figure}

By studying how  the data points $(\beta,\tilde{t})$ cluster together and concentrate or spread compared to the decision region boundary, according to the mismatch level $\cos^2 \theta$ given in \eqref{eq:cos2theta} and signal-to-noise ratio (SNR) defined as
\begin{equation}
\gamma = |\alpha|^2 \bp^{\dagger}  \bC^{-1} \bp \in\R_+
\label{eq:SNR} ,
\end{equation}
several insights on the behavior of the detector were obtained in \cite{TSP_CFAR-FP}. 
In particular, Kelly's horizontal boundary best separates the $H_0$ cluster from any $H_1$ cluster under matched conditions; conversely, detectors with marked robust or selective behaviors 
exhibit an oblique linear or  non-linear boundary, with increasing trend for robust behavior and decreasing trend for selective behavior, as visible in Fig. \ref{fig:FP}, respectively, for AMF and the {\ed robustified GLRT (ROB)} \cite{ROB} and for ACE and WABORT.

Another interesting finding in \cite{TSP_CFAR-FP} is that  $\beta$ 
has a pivotal role in determining the trade-off between performance under matched and mismatched conditions. 
This behavior is observed for instance in  the ROB detector, which has $P_d$ similar to Kelly's detector under matched conditions but is very robust: in fact, its boundary is  increasing  in the lower range of $\beta$ and then becomes constant  (horizontal) for larger values (see again Fig. \ref{fig:FP}).
Another detector that combines the characteristics of two well-known detectors  naturally arises when  a K-nearest neighbors decision rule is adopted  \cite{KNN,KNN2}: the resulting detector has intermediate behavior between Kelly's and AMF, and in fact its decision region boundary approximates a piecewise-linear function close to the positive-slope line of AMF for the lower range of $\beta$ and attains Kelly's horizontal line in the upper range\footnote{The intersection point between the two lines for the case of Fig. \ref{fig:FP} is approximately around $\beta=0.5$.}. 

\subsection{\ed Design Challenges and Motivations}

Unfortunately, a general methodology for the design of customized detectors  (including, but not limited to, combinations of two or more existing detectors) with prescribed behavior in terms of $P_d$ under matched and mismatched conditions  is still missing in the literature.
In the CFAR-FP, this corresponds to choosing an arbitrary decision region boundary of interest 
according to the intended classification of signal points for different SNRs and mismatch levels (clusters) as either $H_0$ or $H_1$. However, the  $P_{fa}$ of the resulting detector cannot be controlled upfront, not even when joining  parts taken from existing detectors having common $P_{fa}$.
Indeed, the relationships between decision region boundary and the induced $P_{fa}$ and $P_d$'s (under matched and mismatched conditions) are highly non-linear; thus, any adjustment around a certain region of the curve (aimed at matching the $P_{fa}$) will have an uneven impact according to the density of points (belonging to the different clusters) that fall in that region of the CFAR-FP ---  so making such an adjustment unintuitive and non-trivial. 
\begin{figure}
\centering
\includegraphics[width=0.4\textwidth]{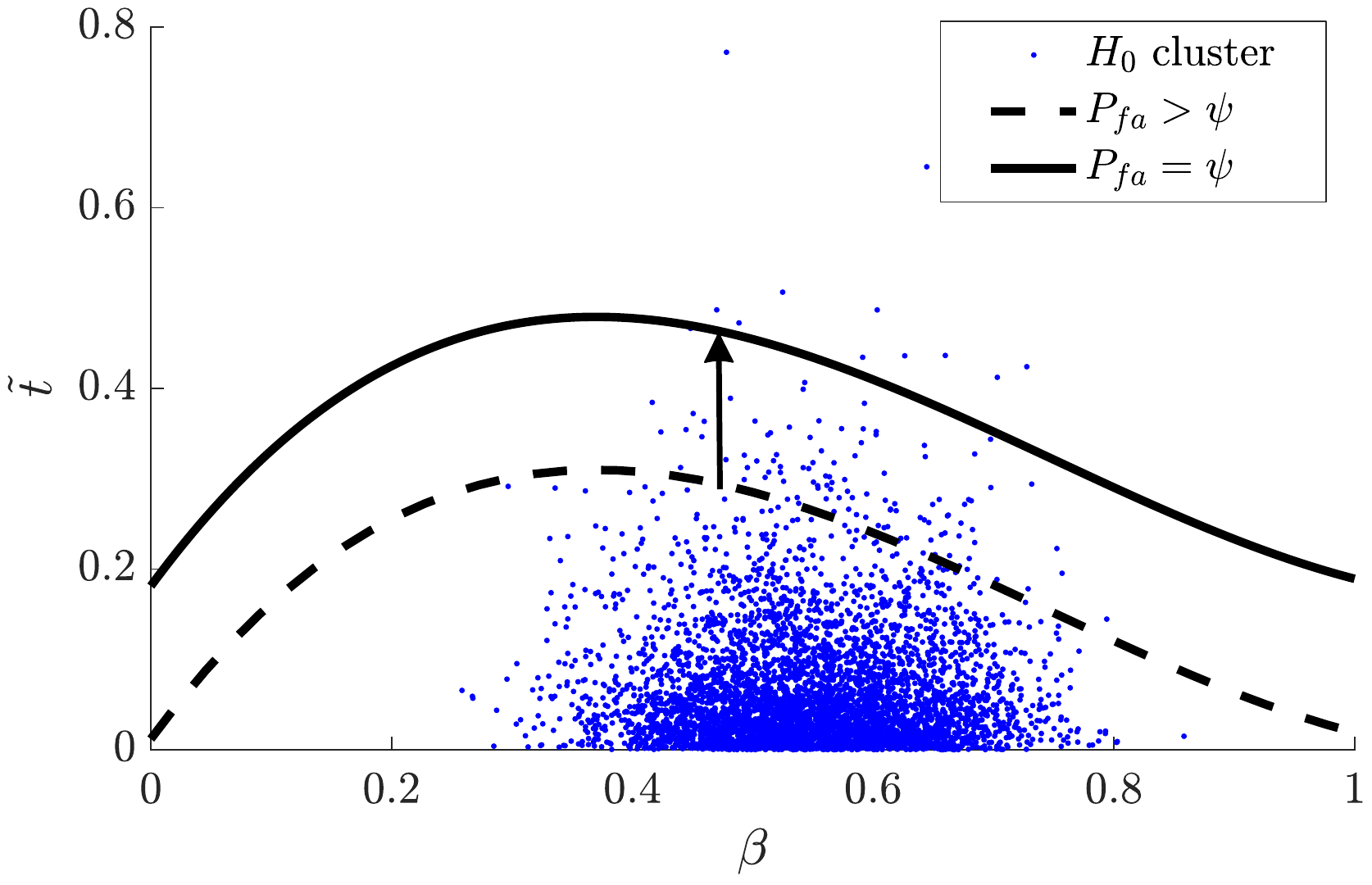}
\caption{\ed Adjustment of  $P_{fa}$  by stiff shift of a decision region boundary.}\label{fig:shift}
\end{figure}
The naive solution adopted in \cite{TSP_CFAR-FP} to adjust the $P_{fa}$ was to stiffly shift the desired curve upwards or downwards, iteratively, while checking $P_{fa}$ and stopping at equality. {\ed As it can be seen from Fig. \ref{fig:shift}, if for instance the desired detector's curve in the CFAR-FP (dashed line) yields a $P_{fa}$ higher than the chosen value $\psi$ --- meaning that the integral of the joint pdf of $(\beta, \tilde{t})$ over the area above the curve exceeds $\psi$ --- the decision region boundary is slightly shifted upwards. This will result in a decrease of $P_{fa}$ since less points of the $H_0$ cluster will fall above the decision boundary. The process is iterated, by shifting the curve upwards or downwards according to the computed value of $P_{fa}$, until a curve is found for which the corresponding $P_{fa}$ exactly matches $\psi$ (solid line).}
Unfortunately, in doing so the performance  in terms of $P_d$ and/or  desired behavior under mismatched conditions will degrade, i.e., the shift may generally jeopardize the design. 

The discussion above motivates the importance of addressing the design of customized detectors in the CFAR-FP under an optimal approximation setup, while guaranteeing the desired $P_{fa}$, as discussed below.

\section{Optimal Design of Customized Adaptive CFAR Radar Detectors}\label{sec::2}

\begin{figure}
\centering
\includegraphics[width=0.42\textwidth]{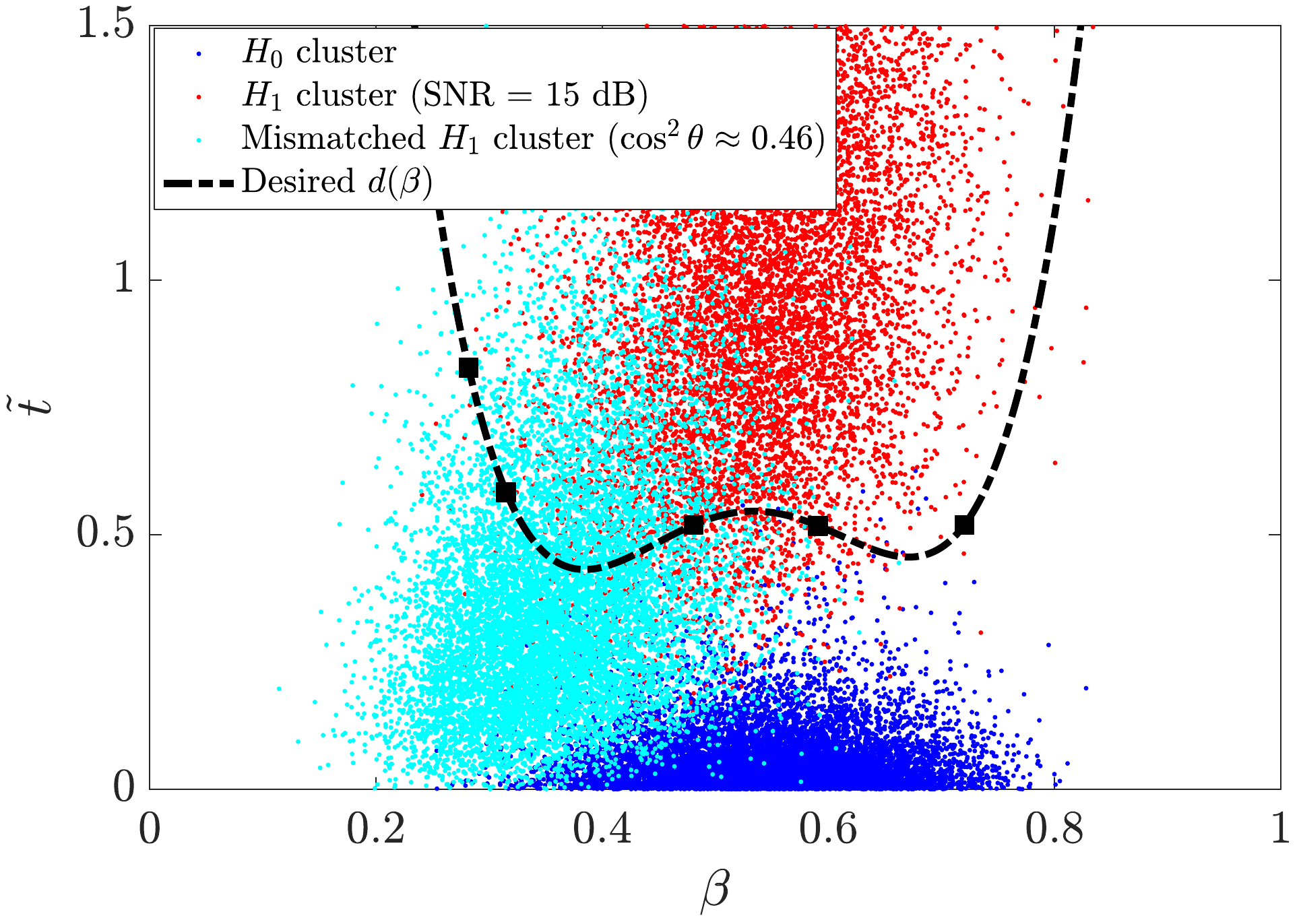}
\caption{Desired decision region boundary $d(\beta)$ of the double-well detector in the CFAR-FP.}\label{fig:ad-hoc}
\end{figure}

\subsection{Preliminary Considerations and Design Examples}\label{sec:examples}

Considering the goal of designing customized detectors with desired properties, while working at a preassigned $P_{fa}$,  we denote by $d(\beta)$ the  decision region boundary of a desired CFAR detector, thus having the form \eqref{eq:test2}. Notice that it is irrelevant how such an  expression is obtained. 
The most general case consists in directly drawing $d(\beta)$ according to an intended classification of the signal points as either $H_0$ or $H_1$ in the different regions of the CFAR-FP. In Fig.~\ref{fig:ad-hoc} we report an example in which the customized detector aims at exhibiting good rejection capabilities of the mismatched signals, but providing at the same time high $P_d$ under matched conditions; given its shape, it will be referred to as ``double-well" detector. {\ed The colored dots represent several realizations of the random variables $(\beta,\tilde{t})$ obtained from Monte Carlo simulation of data vectors $\bz, \bz_1,\ldots,\bz_K$: specifically, blue dots correspond to data generated under $H_0$ (noise only), red dots are the data under $H_1$ for SNR = 15 dB, whereas cyan dots are for $H_1$ under mismatched conditions for $\cos^2 \theta \approx 0.46$.}
The desired $d(\beta)$ is drawn as a fourth-order spline parameterized by five control points (hence it will pass through all of them): the first two  are chosen in correspondence of the upper left-most part of the mismatched $H_1$ cluster, while the remaining three  are  set to better separate the matched $H_1$ cluster from  $H_0$ according to the points dispersion, without including too many points of the mismatched $H_1$ cluster {\ed that fall} underneath $d(\beta)$. 

Special cases of $d(\beta)$ can be obtained by combining two or more decision region boundaries of well-known detectors for non-overlapping intervals of $\beta$, as discussed in Sec. \ref{sec:background},  according to the desired performance. A possibility is to take selected points from the curves of different existing detectors over the domain $\beta \in[0, 1]$, and use them as control points for fitting a (low-order) spline; this will result in a curve with ``intermediate'' characteristics (not necessarily passing through all control points). An example is shown in Fig.~\ref{fig:AMF-ROB}, where the AMF and ROB detectors have been selected and a third-order spline has been used for the fitting. Clearly, it is also possible to simply juxtapose the parts taken from the different detectors, without any interpolation.

\begin{figure}
\centering
\includegraphics[width=0.40\textwidth]{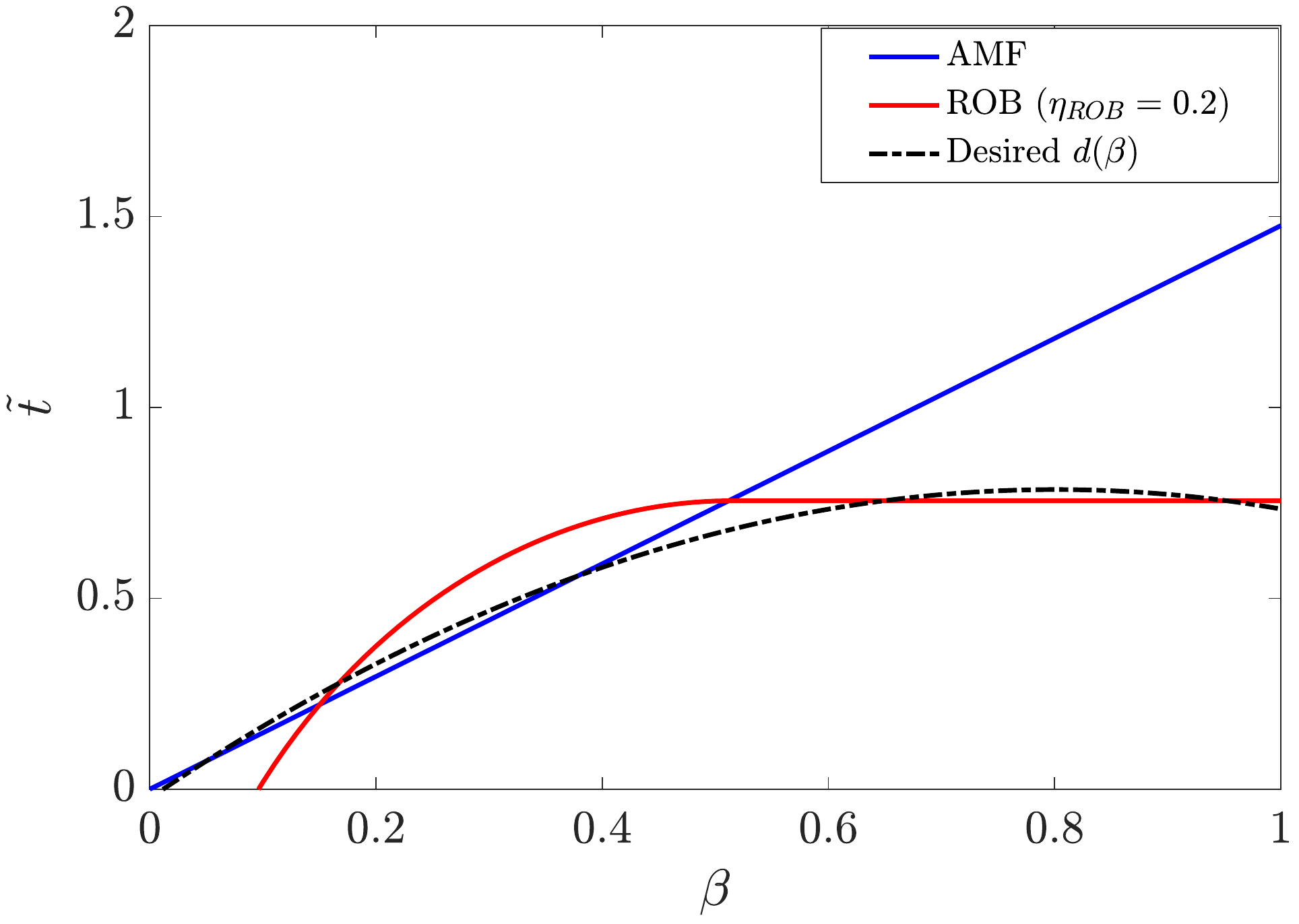}
\caption{Desired decision region boundary $d(\beta)$ obtained as a combination of ROB and AMF detectors.}\label{fig:AMF-ROB}
\end{figure}

\subsection{Optimal Infinite-Dimensional Design Problem}
\label{sec:optimal_infinite}

By construction, $d(\beta)$ corresponds to a CFAR detector with desired performance under matched and mismatched conditions. However, as anticipated, the resulting value of $P_{fa}$ is unpredictable, even in case of combination of detectors with the same $P_{fa}$. {\ed In order to come up with a detector working  at a preassigned $P_{fa}$, which is of fundamental importance in the radar context, a strategy is needed that takes as reference the desired $d(\beta)$ and approximates it through a suitable parametric function until the $P_{fa}$ constraint is exactly fulfilled, while retaining as much as possible its behavior in terms of detection capabilities.}

By taking inspiration from the classical Neyman-Pearson approach, the problem can be  formulated as determining an approximation $f(\beta)$ of $d(\beta)$ that maximizes  $P_d$ under a constraint on $P_{fa}$. We introduce however a more general objective function, in order to take into account also the performance under mismatched conditions, hence in turn to obtain a desired trade-off in this respect. Specifically, we propose to minimize the weighted least squares cost functional
\begin{equation}\label{eq::P1_initial}
\mathcal{C}_1  (f)  =  \sum_{i=1}^S \! w_i(f; \bm{\rho}_1,\ldots,\bm{\rho}_S)\,  e(f; \bm{\rho}_i)
\end{equation}
obtained from a set of $S$ specifications $\bm{\rho}_i = [\gamma_i \ \lambda_i \ \psi_i]^\mathsf{T}$ {\ed (with $(\cdot)^\mathsf{T}$ the transpose operator)} encoding the desired probabilities of detection $\psi_i$ for each chosen pair of SNR $\gamma_i$ and mismatch value $\lambda_i$, with
\begin{equation}
e(f; \bm{\rho}_i) = \left[P(\tilde{t} \! > \!  f(\beta)| \gamma=\gamma_i, \cos^2\! \theta = \lambda_i ) - \psi_i\right]^2
\end{equation}
denoting the squared error between the desired $\psi_i$ and the probability that the detector having decision region boundary $f(\beta)$ will decide for $H_1$ given $\gamma_i$ and $\lambda_i$, {\ed and $w_i$ chosen weighting functions}.
These specifications should be interpreted as ``soft constraints'' that will be not necessarily fulfilled in the solution minimizing \eqref{eq::P1_initial}, given their possible conflicting nature\footnote{In this respect, a similarity can be observed with the problem of filter design as well as beamformer design.}. 
We will discuss in Sec. \ref{sec:proposed} how specification values can be automatically chosen in practice to encode the desired behavior given by $d(\beta)$. Moreover, we will specify a weighting strategy in which $w_i$'s depend on both the function $f(\cdot)$ itself and  the specifications $\bm{\rho}_i, i=1,\ldots,S$. 
 
The infinite-dimensional optimization problem is given by
\begin{equation}\label{eq::inf_dim}
\begin{array}{l}
  \underset{f \in \mathcal{F}}{\mathrm{minimize}} \ \mathcal{C}_1(f) \\
 \text{s.t. } \mathcal{C}_0(f) = P_{fa}
\end{array} 
\end{equation}
where $\mathcal{F}$ is a chosen space of functions defined over $[0,1]$, $\mathcal{C}_1(f)$ is given in  \eqref{eq::P1_initial},
and
\begin{equation}
\mathcal{C}_0(f) = P(\tilde{t} >  f(\beta)| H_0) = 1 - \int_0^1 F_{\tilde{t}|H_0} (  f(\beta))p(\beta) d\beta \label{eq:C0f}
\end{equation}
with $F_{\tilde{t}|H_0}(\cdot)$ denoting the cumulative distribution function (CDF) of $\tilde{t}$  under the $H_0$ hypothesis, and $p(\cdot)$ denoting the pdf of the complex central Beta distribution {\ed (better discussed in Sec. \ref{sec:algorithm})}.
{\ed The exact solution of such an infinite-dimensional optimization problem would yield the curve that best approximates the desired $d(\beta)$ in the CFAR-FP, minimizing at the same time the deviation from the desired behavior expressed through the $S$ specifications. Unfortunately, it involves transcendental functions in both objective function and constraint}, making its analytical resolution a formidable task. Indeed, next Sec. \ref{sec:proposed} is devoted to the development of a low-complexity approach based on a finite-dimensional version of \eqref{eq::inf_dim}, which will however require  to devise a suitable search strategy given the highly non-convex nature of the optimization problem at hand.

\subsection{Optimal Finite-Dimensional Design Problem}\label{sec:optimal_finite}

We propose a more practical approach which consists in solving a finite-dimensional version of the (intractable) optimization problem \eqref{eq::inf_dim}, where the goal is to approximate the desired decision region boundary $d(\beta)$ through a parametric curve $f(\beta; \bm{x})$, with $\bm{x}$ a vector of real parameters.
The optimal $\bm{x}$ that guarantees a preassigned $P_{fa}$ and minimizes the cost function is obtained by solving the optimization problem
\begin{equation}
\begin{array}{l}
\underset{\bm{x}}{\mathrm{minimize}}
 \ C_1(\bm{x}) \\
 \text{s.t. } C_0(\bm{x}) = P_{fa}
\end{array} \label{eq:opt}
\end{equation}
where
\begin{equation}
C_1(\bm{x}) =  \sum_{i=1}^S \! w_i(\bm{x}; \bm{\rho}_1,\ldots,\bm{\rho}_S) e(\bm{x}; \bm{\rho}_i) \label{eq:P1}.
\end{equation}
Notice that, besides depending on the specifications $\bm{\rho}_i, i=1,\ldots,S$, the weights $w_i(\cdot)$ may also depend upon  the parametric curve $f(\beta; \bm{x})$, as in \eqref{eq::P1_initial}, but through the optimization {\ed vector} $\bm{x}$. The quadratic loss
\begin{equation}
e(\bm{x}; \bm{\rho}_i) = \Big[ P(\tilde{t} \!> \! f(\beta; \bm{x})| \gamma=\gamma_i, \cos^2\!\theta = \lambda_i ) - \psi_i\Big]^2
\end{equation}
has been chosen as error function and, likewise \eqref{eq:C0f},
\begin{equation}
C_0(\bm{x}) = P(\tilde{t} >  f(\beta; \bm{x})| H_0) = 1 - \int_0^1 F_{\tilde{t}|H_0} ( f(\beta; \bm{x})) p(\beta) d\beta. \label{eq:P0}
\end{equation}
Resolution of the problem above requires to specify the parametric function $f(\beta; \bm{x})$.
In the following,  we propose a convenient structure for {\ed the latter, which leads to a convenient analytical characterization of the involved statistics}. Based on that, a novel reduced-complexity algorithm is devised for designing customized detectors in the CFAR-FP, according to the optimization problem \eqref{eq:opt}.

\section{Proposed Resolution Approach}\label{sec:proposed}

\subsection{Choice of Parametric Function $f(\beta; \bm{x})$}\label{sec:choice}

We propose the adoption of a piece-wise structure for the parametric function $f(\beta; \bm{x})$, as follows:
\begin{equation}
f_{\bm{m}}(\beta;\bm{\epsilon}) = \sum_{i=1}^p f_{m_i}(\beta; \epsilon_i) \, \Pi \left( \frac{\beta- i/p + 1/(2p)}{1/p} \right) \label{eq:structure}
\end{equation}
where for simplicity (and without loss of generality\footnote{The proposed approach can be straightforwardly extended to the case of non-uniform partition of the interval $[0,1]$, which could accommodate tighter approximations in certain regions and looser approximations in other ones, according to the $d(\beta)$ at hand.}) we have considered a uniform partition of the domain $[0,1]$ in which $\beta$ takes values, $\Pi(\cdot)$ is the rectangular window centered in the origin with unitary amplitude over $[-1/2,1/2]$ (and zero elsewhere), and  $\{f_{m_i}(\beta;\epsilon_i), i=1, \ldots, p\}$ is a set of elementary functions to be used in the approximation of $d(\beta$) according to  \eqref{eq:opt}.
Among the different alternatives,
the simplest one is to adopt a piecewise-linear approximation, i.e.,
\begin{equation}
    f_{m_i}(\beta; \epsilon_i) = m_i \beta + \epsilon_i.
\end{equation}
Notice that in general \eqref{eq:structure} depends on $2p$ parameters to be optimized, i.e., $\bm{x}=[\bm{m}^\mathsf{T} \ \bm{\epsilon}^\mathsf{T}]^\mathsf{T}$, where $\bm{m}=[m_1 \ \cdots \ m_p]^\mathsf{T}$ and $\bm{\epsilon}=[\epsilon_1 \ \cdots \ \epsilon_p]^\mathsf{T}$.

The resolution of the optimization problem \eqref{eq:opt}  will generally lead to a decision region boundary that is discontinuous.
If one is interested in having a continuous (piecewise-linear) solution, the optimization problem can be easily extended by adding the following continuity constraint:
\begin{equation}\label{eq::cont_constr}
 [\bm{A} \ \bm{b} \ \bm{C} \ \bm{d}]\begin{bmatrix} \bm{m} \\
\bm{\epsilon}
\end{bmatrix} = \bm{0}_{(p-1) \times 1} 
\end{equation}
where $\bm{A}$ is a $(p-1)$-dimensional bidiagonal matrix with diagonal elements $i/p$ {\ed ($i=1,\ldots,p-1$) and upper diagonal elements $-i/p$ ($i=1,\ldots,p-2$)}, i.e., 
\begin{equation}
\bm{A} = \left[
\begin{array}{ccccc}
\frac{1}{p} & -\frac{1}{p}  & 0 & \cdots & 0 \\
0 & \frac{2}{p} & -\frac{2}{p} & \ddots & \vdots \\
\vdots & 0 & \ddots & \ddots & 0  \\
\vdots & \vdots & \ddots & \ddots & -\frac{p-2}{p} \\
0 & \cdots  &  \cdots & 0 & \frac{p-1}{p} 
\end{array} \right] 
\end{equation}
and, analogously, $\bm{C}$ is a $(p-1)$-dimensional bidiagonal matrix with diagonal elements $1$ and upper diagonal elements $-1$, while $\bm{b}$ and $\bm{d}$ are $(p-1)$-dimensional vectors with all-zero elements except for the last one, equal to $-\frac{p-1}{p} $ and $-1$, respectively.

Unfortunately, \eqref{eq:opt} is highly non-convex and local minima can be abundant, hence the numerical resolution of this problem (with or without continuity constraint) is troublesome\footnote{Even state-of-the-art global solvers typically fail in this task. In our trials, we have used an interior-point algorithm initialized with a scatter-search mechanism for generating multiple start points \cite{InteriorPoint}. We have also tested  other state-of-the-art global optimization techniques such as direct methods (e.g., Pattern Search \cite{PatternSearch1, PatternSearch2}) and genetic algorithms \cite{GeneticAlgorithm}. The results in all our trials, under different conditions and parameter settings, invariably lead to unfeasible solutions or very poor local minima, far from being an acceptable approximation of $d(\beta)$.}.
To overcome such difficulties, we propose a different optimization strategy that seeks for a feasible solution in a limited, but sufficiently rich subset of the solution space, as detailed below.

Our starting point is to reduce the parameter space from $2p$ to $p$, by keeping fixed the $p$ parameters in $\bm{m}$  while optimizing the $p$ parameters in $\bm{\epsilon}$.
The key aspect of this choice is that on each interval $[(i-1)/p, i/p]$ the approximant function will depend upon a single parameter, i.e., the affine term $\epsilon_i$. 
In particular, we consider a discontinuous piecewise-linear approximation obtained by  juxtaposition of the best linear fitting of $d(\beta)$ in each interval $i$, as shown in the example reported in Fig.~\ref{fig:piecewise}, and better discussed later. 
\begin{figure}
\centering
\includegraphics[width=0.45\textwidth]{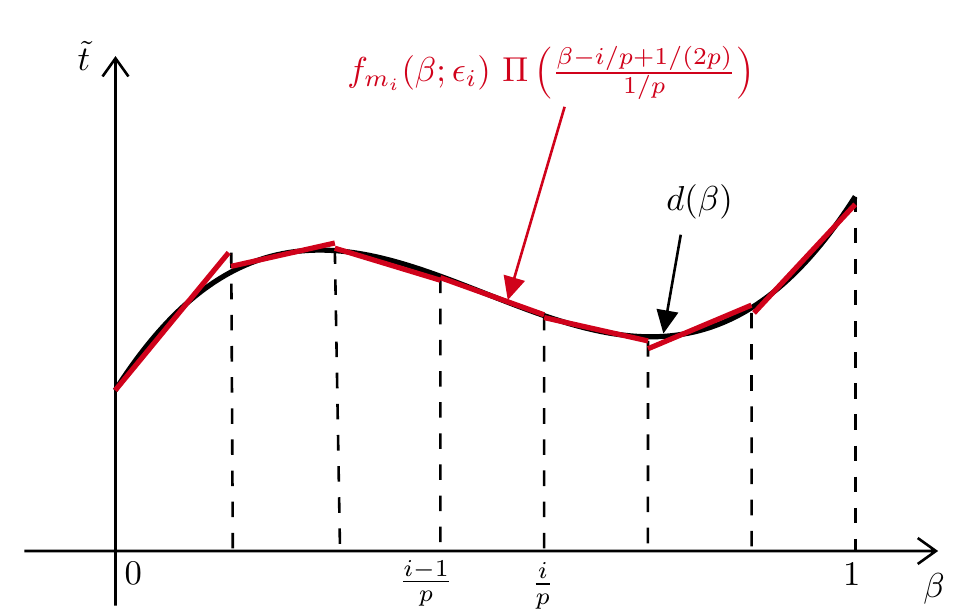}
\caption{Example of \eqref{eq:structure} for $f_{m_i}(\beta; \epsilon_i) = m_i\beta + \epsilon_i$ and $p=7$.}\label{fig:piecewise}
\end{figure}
Following this sub-optimal optimization strategy, we develop  a practical algorithm that exploits the piecewise-linear structure of $f_{\bm{m}}(\beta; \bm{\epsilon})$ to end up with a low-complexity resolution approach. In doing so, we are able to provide a satisfactory approximated solution of the (intractable) original optimization problem, optionally also with continuous boundary.

To start with, eqs. \eqref{eq:P1}-\eqref{eq:P0} can be made more explicit by exploiting the structure in \eqref{eq:structure}; specifically, for any value of SNR (including $\gamma=0$, i.e., $H_0$) and mismatch level (including $\cos^2\!\theta=1$, i.e., $H_1$ under matched conditions), we can write
\begin{align}
 P(\tilde{t} \!> \! f_{\bm{m}}(\beta; \bm{\epsilon})| \gamma, \cos^2\!\theta )& = 1 - \int_0^1 F_{\tilde{t}|\beta} ( f_{\bm{m}}(\beta;\bm{\epsilon})) p(\beta)  d\beta \nonumber \\
&= 1\! -\! \sum_{i=1}^p r_i (\epsilon_i, \gamma)
\end{align}
where each term
\begin{equation}
r_i (\epsilon_i, \gamma) = \int_{(i-1)/p}^{i/p} \!\! F_{\tilde{t}|\beta}(f_{m_i}(\beta; \epsilon_i)) p(\beta)d\beta \label{eq:r_i} 
\end{equation}
is a monotonically increasing one-dimensional function in $\epsilon_i$. Using these new expressions, the optimization problem can be finally recast as
\begin{equation}
\begin{array}{l}
\underset{\bm{\epsilon}}{\mathrm{minimize}}
 \ C_1(\bm{\epsilon}) \\
 \text{s.t. } C_0(\bm{\epsilon}) = P_{fa}
\end{array} \label{eq:final_opt}    
\end{equation}
where
\begin{equation}\label{eq::finalcost}
C_1(\bm{\epsilon}) = \sum_{k=1}^S e(\bm{\epsilon};\bm{\rho}_k)  \sqrt{ \sum_{i=1}^S\left( e(\bm{\epsilon};\bm{\rho}_i) - \frac{1}{S} \sum_{j=1}^S e(\bm{\epsilon};\bm{\rho}_j) \right)^2}
\end{equation}
with
$$
e(\bm{\epsilon};\bm{\rho}_i) = \Big[ P(\tilde{t} \!> \! f_{\bm{m}}(\beta; \bm{\epsilon})| \gamma=\gamma_i, \cos^2\!\theta = \lambda_i ) - \psi_i\Big]^2
$$
and 
\begin{equation}\label{eq::finalconstr}
C_0(\bm{\epsilon}) =    P(\tilde{t} >  f_{\bm{m}}(\beta; \bm{\epsilon})| H_0) = 1 - \sum_{i=1}^p r_i(\epsilon_i,0).
\end{equation}
The adopted cost function \eqref{eq::finalcost} is one among several possible choices, and is aimed at promoting a fairly uniform deviation of the approximated decision region boundary from the desired one, while minimizing the overall approximation error. Details about its derivation and a discussion of alternative choices are reported in Appendix~\ref{app:discussion}.

We now provide two key Propositions where the peculiar structure of $f_{\bm{m}}(\beta; \bm{\epsilon})$ is exploited together with the statistical characterization of the maximal invariant statistics $(\beta,\tilde{t})$ to derive i) a more compact formula to compute the integral in \eqref{eq:r_i} under the $H_1$ (matched/mismatched) hypothesis, which will be used to evaluate the cost function in \eqref{eq::finalcost}; ii) a closed-form solution to the integral appearing in \eqref{eq:r_i} under the $H_0$ hypothesis, which will be exploited to compute the $P_{fa}$ constraint in \eqref{eq::finalconstr}. Based on such results, we will be able to devise a novel reduced-complexity algorithm for the design of customized detectors according to problem \eqref{eq:final_opt}.

\subsection{Analytical Characterization}\label{sec:algorithm}

First recall the general characterization of $(\beta,\tilde{t})$ parameterized in $\gamma$ and $\cos^2 \theta$, which encompasses the one under $H_0$ (for $\gamma=0$) and $H_1$ under matched conditions (for $\cos^2 \theta=1$) \cite{Kelly_techrep}, see also \cite{Kelly,Kelly89,BOR-MorganClaypool,KellyTR2,Richmond}.
The variable $\tilde{t}$ given $\beta$ is ruled by a complex noncentral $\mathcal{F}$-distribution with $1$ and $K-N+1$ complex degrees of freedom and noncentrality parameter $\gamma \beta \cos^2\theta$, i.e., $\tilde{t}\sim \mathcal{CF}_{1,K-N+1}(\gamma \beta \cos^2\theta)$; $\beta$ is ruled by a complex noncentral Beta distribution with $K-N+2$ and $N-1$ complex degrees of freedom and noncentrality parameter $\gamma (1-\cos^2\theta)$, i.e., $\beta\sim \mathcal{C\beta}_{K-N+2,N-1}(\gamma (1-\cos^2\theta))$. We now provide the following results.

\begin{proposition}\label{Prop1}
Consider for $f_{\bm{m}}(\beta; \bm{\epsilon})$ in \eqref{eq:structure} the set of affine functions $f_{m_i}(\beta; \epsilon_i) = m_i\beta + \epsilon_i$,  $i=1,\ldots,p$; then,  $r_i(\epsilon_i,\gamma)$ in eq. \eqref{eq:r_i} can be more conveniently computed as
\begin{align}\label{ri_generic}
r_i(\epsilon_i,\gamma) = \int_{(i-1)/p}^{i/p} \Psi(f_{m_i}(\beta;\epsilon_i)) \Omega(\beta)d\beta 
\end{align}
where
\begin{align*}
\Psi & (f_{m_i}(\beta; \epsilon_i))  = (1+f_{m_i}(\beta; \epsilon_i))^{-(K-N+1)} \nonumber \\
&\times \sum_{\ell=1}^{K-N+1}\binom{K-N+1}{\ell} \frac{f_{m_i}(\beta;\epsilon_i)^\ell}{\Gamma(\ell)} \Gamma\left(\ell,\frac{\delta^2_F}{1+f_{m_i}(\beta; \epsilon_i)}\right)
\end{align*}

\begin{align*}
\Omega(\beta)  =  \binom{K}{N-2} & \frac{\e^{-\delta^2_\beta}\beta^{K-N+1}}{(1-\beta)^{2-N}}\\ & \times {}_1F_1(-K+N-2,N-1,\delta_\beta^2(\beta-1))
\end{align*}
with $\Gamma(a,b)$ the Euler's upper incomplete Gamma function, ${}_1F_1(a,b;z)$ the Kummer's confluent hypergeometric function, $\delta_F = \gamma\beta\cos^2 \theta$ the noncentrality parameter of the $F$-distribution, and $\delta_\beta = \gamma(1-\cos^2 \theta)$ the noncentrality parameter of the complex Beta distribution.
\end{proposition}
\begin{proof}
See Appendix \ref{app:A}, where also the following Corollary is obtained as a by-product of the proof.
\end{proof}
\begin{corollary}
For negative integer $a$ and  positive integer $b$,  ${}_1F_1(a,b;z)$ is a polynomial of degree $-a$ \cite[eq. 13.1.3]{abramowitz}, hence an alternative expression for $\Omega(\beta)$ is
\begin{align*}
\Omega(\beta) = \frac{\e^{-\delta^2_\beta}\beta^{K-N+1}}{(1-\beta)^{2-N}} L_{K-N+2}^{(N-2)}(\delta_\beta^2(\beta-1))
\end{align*}
where $L^{(\alpha)}_n(x)$ is the generalized Laguerre polynomial of order $n$ and parameter $\alpha$, here computed in $x=\delta_\beta^2(\beta-1)$.
\end{corollary}

\begin{proposition}\label{Prop2}
Under the $H_0$ hypothesis ($\gamma = 0$), the integral appearing in \eqref{ri_generic} can be solved in closed-form as
\begin{align}\label{eq:ri_Pfa}
r_i&(\epsilon_i,0) = F_{\beta|H_0}(i/p) - F_{\beta|H_0} ((i-1)/p) \nonumber \\
& - \frac{K!}{(K-N+1)!(N-2)!}\left[g(i/p,\epsilon_i) - g((i-1)/p,\epsilon_i) \right]
\end{align}
with $F_{\beta|H_0}(\cdot)$ denoting the CDF of the complex central Beta distribution and
\begin{align*}
&g(x,\epsilon_i) = \frac{x}{K-N+2} \left(\frac{x}{1+\epsilon_i}\right)^{K-N+1} \nonumber \\ & \times \! F_1 \! \left(\! K\! - \! N \! +\! 2,2-N,K-N+1,K-N+3; x, \frac{-m_ix}{(1+\epsilon_i)} \right)
\end{align*}
with $F_1(a,b,c,d; y,z)$ the Appell $F_1$ function of two variables.
\begin{proof}
See Appendix \ref{app:B}.
\end{proof}
\end{proposition}

\subsection{\ed Low-Complexity Design Procedure}

The main challenge with problem \eqref{eq:final_opt} is that the joint optimization of the $p$ parameters in $\bm{\epsilon}$ is still non-trivial as both the cost function and the constraint
encode a highly non-linear dependency on the vector $\bm{\epsilon}$. To circumvent this challenge, we pursue an alternative path that seeks for a feasible solution of \eqref{eq:final_opt} by iteratively exploring a range of piecewise-linear approximations of $d(\beta)$. The algorithm takes as inputs the desired decision region boundary $d(\beta)$, the maximum dimension $p$ of the parameter vector $\bm{\epsilon}$, and the desired $P_{fa}$, and returns in output the customized piecewise-linear detector $f_{\bm{m}}(\beta; \bm{\epsilon}^*)$ working at the preassigned $P_{fa}$.

\begin{figure}
\includegraphics[width=0.48\textwidth]{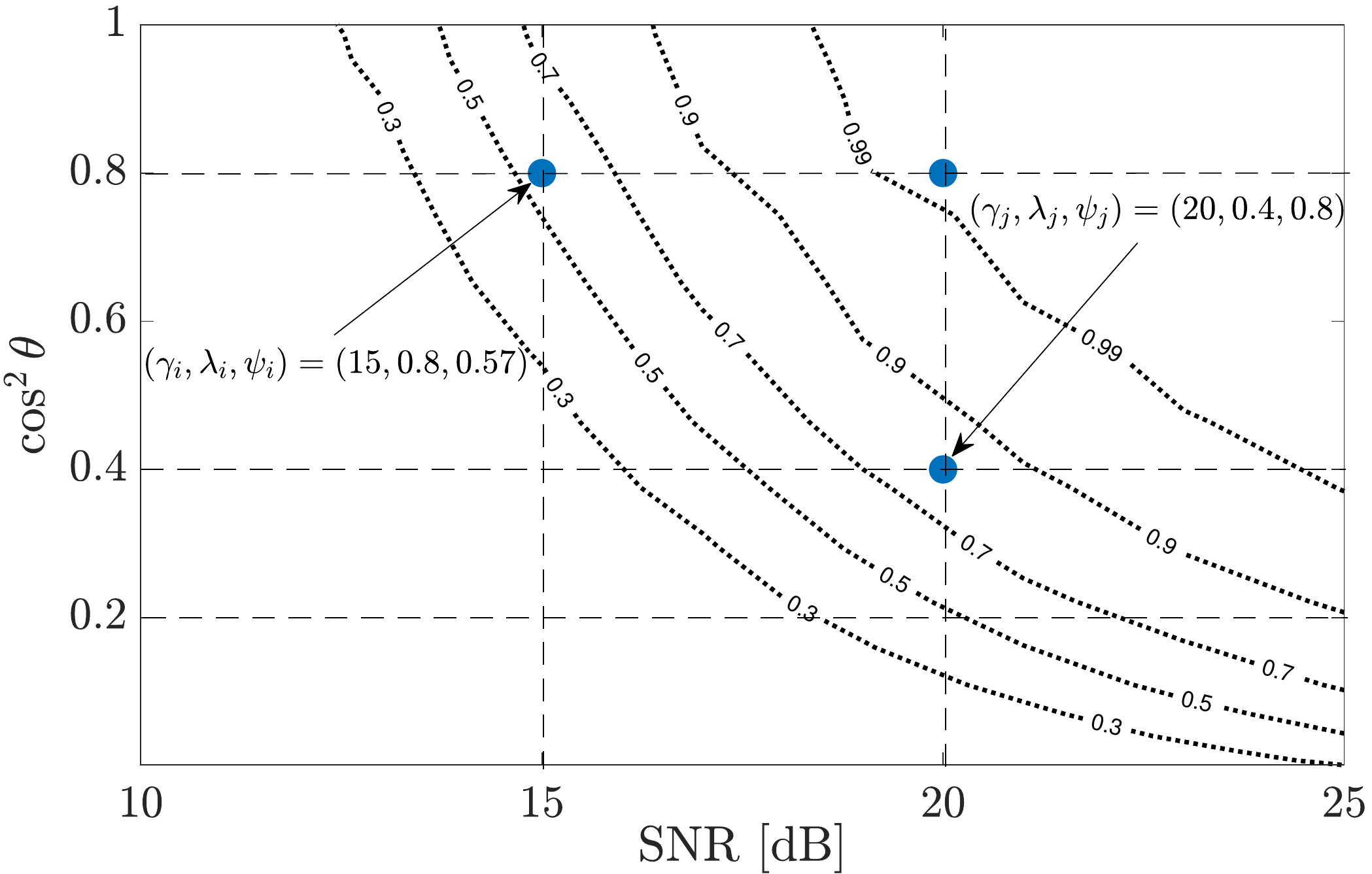}
\caption{Proposed strategy for the automatic setting of the $S$ specifications.}\label{fig:mesa}
\end{figure}

The specifications $\bm{\rho}_i, i=1,\ldots,S$, have to be chosen in order to correctly encode the desired detection performance under matched and mismatched conditions. To this end, we propose an automatic approach in which the $S$ specifications are directly obtained by sampling  the mesa-plots of $d(\beta)$, i.e., each $\bm{\rho}_i$ is set to a point lying on an iso-$P_d$ curve (with level $\psi_i$) in the SNR-$\cos^2 \theta$ plane (with coordinates $(\gamma_i,\lambda_i$)), as shown in Fig.~\ref{fig:mesa}.\footnote{\ed Notice that although the parameters in $\bm{\rho}_i$ are theoretically defined as function of the true $\bR$, their values are chosen from the curves in the mesa-plot, which do not require any knowledge of $\bR$ or other parameters.}
This choice has the advantage of not requiring any ad-hoc setting of the involved parameters. Moreover, it appears a natural way to make the design requirements in the cost function $C_1(\bm{\epsilon})$ compatible with the desired $d(\beta)$, facilitating the algorithm in finding a feasible solution $f_{\bm{m}}(\beta;\bm{\epsilon})$ that works at the preassigned $P_{fa}$ value. We will refer to such a procedure by means of the functional notation $(\bm{\rho}_1, \ldots, \bm{\rho}_S) = \mathrm{SampleSpecifications}(d(\beta),S)$.

Once the specifications have been set, the proposed algorithm performs  $\sum_{k=2}^{p} k = \frac{p(p+1)}{2} - 1$ iterations in which the function $f_{\bm{m}}(\beta;\bm{\epsilon})$ is progressively reparameterized by an increasing number of optimization variables $k$, ranging from $k=2$ up to $k = p$. According to the  results in Proposition \ref{Prop1} and \ref{Prop2}, we consider for $f_{\bm{m}}(\beta; \bm{\epsilon})$ the set of affine functions $f_{m_i}(\beta; \epsilon_i) = m_i\beta + \epsilon_i$,  $i=1,\ldots,k$, where the parameters $m_i$ and $\epsilon_i$ are initially set equal to the coefficients of the line that approximates (in least squares sense) the desired $d(\beta)$ for $\beta \in [(i-1)/k, i/k]$. In other words, the proposed method explores a range of piecewise-linear approximations of $d(\beta)$ from a coarse ($k = 2$) up to a fine scale ($k = p$).

For each $k$-dimensional parameterization of $f_{\bm{m}}(\beta; \bm{\epsilon})$, our strategy consists in iteratively changing only a single affine parameter $\epsilon_i$ at a time, while keeping fixed the remaining $k-1$ elements of the optimization vector $\bm{\epsilon}$. More specifically, we consider the decomposition of \eqref{eq::finalconstr} as
$$
C_0(\bm{\epsilon}) = 1 - r_i(\epsilon_i,0) - \sum_{j=1, j\neq i}^k r_j(\epsilon_j,0).
$$
The modified value will then correspond to the root of the equation $C_0(\bm{\epsilon}) - P_{fa} = 0$ solved with respect to $\epsilon_i$ using the result in Proposition \ref{Prop2}, that is, the algorithm attempts to modify  $\epsilon_i$ in order to exactly fulfill the $P_{fa}$ constraint. Among all the  configurations of the parametric function $f_{\bm{m}}(\beta; \bm{\epsilon})$ satisfying the $P_{fa}$ constraint, we retain as best approximation of $d(\beta)$ the  piecewise-linear detector $f_{\bm{m}}(\beta; \bm{\epsilon}^*)$ with $k^*$ parameters leading to the minimum value of the cost function $C_1(\bm{\epsilon})$ in \eqref{eq::finalcost}, the latter evaluated using the result in Proposition~\ref{Prop1}. The steps of the proposed approach are summarized in Algorithm \ref{alg}.

\begin{algorithm}\label{alg}
\DontPrintSemicolon
  \KwInput{$d(\beta)$, $p$, $P_{fa}$ }
  \KwOutput{$k^*$, $f_{\bm{m}}(\beta; \bm{\epsilon}^*)$}\vspace{0.2cm}
  ($\bm{\rho}_1$, \ldots, $\bm{\rho}_S$) = $\mathrm{SampleSpecifications}(d(\beta),S)$ \\
  \For{$k = 2$ to $p$}{
  \For{$i = 1$ to $k$}{
  Set $f_{m_i}(\beta;\epsilon_i)$ in eq.~\eqref{eq:structure}:
  $$
  f_{m_i}(\beta; \epsilon_i) \leftarrow  m_i\beta + \epsilon_i 
  $$
  with $m_i$, $\epsilon_i$ obtained from the linear fitting (regression) of $d(\beta)$ over $\beta \in [(i-1)/k, i/k]$
  }
  \For{$i = 1$ to $k$}{

  Compute $\sum_{j=1, j\neq i}^k r_j(\epsilon_j,0)$ using  \eqref{eq:ri_Pfa} \;
  
  Solve $C_0(\bm{\epsilon}) - P_{fa} = 0$ wrt $\epsilon_i$:
  $$
  \epsilon_i \leftarrow r_i^{-1} \Big(1  - P_{fa} - \sum_{j=1, j\neq i}^k r_j (\epsilon_j,0)\Big)
  $$
  
  Evaluate $C_1(\bm{\epsilon})$ in \eqref{eq::finalcost} using \eqref{ri_generic}

  }
  }
  Select $[k^*, f_{\bm{m}}(\beta; \bm{\epsilon}^*)]$ providing minimum $C_1(\bm{\epsilon})$
  
\caption{Proposed reduced-complexity algorithm}
\end{algorithm}

Intuitively, the proposed approach consists in deforming only a small portion of the desired decision region boundary $d(\beta)$, but at the same time considering a range of piecewise-linear approximations from coarse to fine scale. This captures the intrinsic trade-off between goodness of the approximation (while fulfilling exactly  $P_{fa}$)  and deviation from the desired performance, the latter expressed through the specifications encoded in the cost function. 

\begin{algorithm}\label{alg_continuous}
\DontPrintSemicolon
  \KwInput{$k^*$, $f_{\bm{m}}(\beta; \bm{\epsilon}^*)$}
  \KwOutput{Continuous version of $f_{\bm{m}}(\beta; \bm{\epsilon}^*)$}\vspace{0.2cm}
  
  Initialize $a_0 \leftarrow  \epsilon_1$, $a_{k^*} \leftarrow m_{k^*} + \epsilon_{k^*}$  
  
  \For{$i = 1$ to $k^*-1$}{
  Compute mid-point $a_i \leftarrow (m_{i} \frac{i}{k^*} + \epsilon_{i} + m_{i+1} \frac{i}{k^*} + \epsilon_{i+1})/2$\;
  
  Compute $\delta_i \leftarrow a_{i}-a_{i-1}$
  
  Update $m_i \leftarrow k^* \delta_i$
  
  Update $\epsilon_i \leftarrow -i \delta_i + a_i$ 
  
  }

Update $\epsilon^*_i \leftarrow \epsilon^*_i + \epsilon, \forall i$, with $\epsilon\in\R$ such that $C_0(\bm{\epsilon}^*) - P_{fa} = 0$ is satisfied

\caption{Refinement stage for obtaining a continuous decision region boundary}
\end{algorithm}

The proposed algorithm yields, by construction, a (mildly) discontinuous decision region boundary.
If one is interested in enforcing an exactly continuous boundary, a further adjustment can be performed, as summarized in Algorithm~\ref{alg_continuous}. Each edge of the partition described by the $k^*$ segments returned by Algorithm~\ref{alg} can be easily made continuous by joining  two adjacent segments in their middle point $a_i$
, given by 
\begin{equation}
 a_i = \left(m_{i} \frac{i}{k^*} + \epsilon_{i} + m_{i+1} \frac{i}{k^*} + \epsilon_{i+1}\right)/2
\end{equation}
for $i=1,\ldots, k^*-1$.
The two extreme points corresponding to $\beta = 0$ and $\beta = 1$ are instead kept fixed to their values, that is, $a_0 = \epsilon_1$ and $a_{k^*} = m_{k^*}+\epsilon_{k^*}$.
The resulting continuity-adjusted boundary will have, in each segment, parameters $m_i$ and $\epsilon_i$ modified according to the equation of the straight line passing through the points $((i-1)/k^*,a_{i-1})$ and $(i/k^*,a_i)$, i.e.,
\begin{equation}
m_i = \frac{a_{i}-a_{i-1}}{1/k^*} = k^*(a_i-a_{i-1})
\end{equation}
and 
 \begin{equation}
\epsilon_i = \frac{i/k^* a_{i-1} - (i-1)/k^* a_i}{1/k^*} = i (a_{i-1}-a_i) + a_i
\end{equation} 
for $i=1,\ldots, k^*$, so returning a continuous piecewise-linear version of 
$f_{\bm{m}}(\beta; \bm{\epsilon}^*)$.
Clearly, this refinement stage will (slightly) violate the $P_{fa}$ constraint; however, such a deviation is minor and can be safely recovered by a final vertical shift of the whole curve, until the constraint is exactly satisfied (step 8 in Algorithm~\ref{alg_continuous}). Results in the next Sec. \ref{sec::NumResults} will show that this procedure, given its minimal impact, does not produce any appreciable performance degradation.

For completeness, in Algorithm~\ref{alg2} we explicitly report the decision rule of the piecewise-linear customized detector $f_{\bm{m}}(\beta; \bm{\epsilon}^*)$ for a given realization of the maximal invariant statistics $(\beta,\tilde{t})$. {\ed As for the parameters $(m_\ell, \epsilon^*_\ell)$, they correspond to a specific pair selected among the $k^*$ pairs $(m_i, \epsilon^*_i)$, $i=1,\ldots,k^*$, that constitute $f_{\bm{m}} (\beta; \bm{\epsilon}^*)$.}
The detector simply finds the specific interval $[(\ell-1)/k^*,\ell/k^*]$ in which the $\beta$ statistic falls, and then uses the corresponding parameters $(m_\ell$, $\epsilon^*_\ell)$ to test whether the $\tilde{t}$ statistic exceeds the corresponding threshold $\eta_\ell = m_\ell\beta + \epsilon^*_\ell$ (decide for $H_1$) or not (decide for $H_0$). A convenient interpretation of the detection rule can be visualized in the CFAR-FP as testing whether the $\tilde{t}$ statistic falls above or below the line with parameters $(m_\ell$, $\epsilon^*_\ell)$, which represents the portion of the decision region boundary to be considered for that specific realization of $\beta$.

\begin{algorithm}\label{alg2}
\DontPrintSemicolon
  \KwInput{$\tilde{t}$, $\beta$, $k^*$, $f_{\bm{m}}(\beta; \bm{\epsilon}^*)$ }
  \KwOutput{Decision for $H_0$ or $H_1$}\vspace{0.2cm}
Find index $\ell$ such that $\beta \in [(\ell-1)/k^*,\ell/k^*]$\;
Detection rule:
$$
\tilde{t}  \test m_\ell\beta + \epsilon^*_\ell
$$
\caption{Detection rule of the customized detector}
\end{algorithm}

\section{Performance Assessment}\label{sec::NumResults}
In this section, we assess the performance of the two examples of radar detectors discussed in Sec.~\ref{sec::2}, whose design is obtained through the methodology proposed in Sec.~\ref{sec:proposed}. {\ed The analysis is performed on both simulated and real data.}
The design is conducted assuming  a maximum number of segments for  partitioning the domain $\beta \in[0, 1]$ set to $p = 16$.

Thresholds are set by Monte Carlo simulation with $100/P_{fa}$ independent trials.
For the target simulation, we assume $\bv=[1\  \e^{j2\pi f_d} \ \cdots \ \e^{j2\pi (N-1)f_d} ]^\mathsf{T}$ with a normalized Doppler frequency $f_d=0.08$ (a small value such that the target competes with low pass clutter). To model a mismatched target, we define the actual steering vector $\bp$ as $\bv$ but with $f_d+\delta_f$ and $\delta_f$ varying in order to obtain different levels of mismatches.

\subsection{Performance of the Double-well Detector}\label{sec::double_well}

{\ed
We start by analyzing the first example of detector design presented in Sec.~\ref{sec::2},  labeled ``double-well''.  
We consider $N=16$, $K=32$, and a  desired $P_{fa} = 10^{-4}$. 
As to $\bC$, we consider the sum of a Gaussian-shaped clutter  and white (thermal) noise 10 dB weaker, i.e.,
$\bC =  \bR_c + \sigma_n^2 \bI_N$ with  the $(m_1,m_2)$th element of the matrix $\bR_c$ given by
$[\bR_c]_{m_1,m_2} = \exp\{- 2\pi^2\sigma_f^2(m_1-m_2)^2\}$ and $\sigma_f \approx 0.051$ (corresponding to a one-lag correlation coefficient of the clutter component equal to $0.95$).
Finally, we consider $10^3$ independent trials to compute the $P_ d$s. }

The  $d(\beta)$ of Sec.~\ref{sec::2} has a $P_{fa} \approx 10^{-3}$, which is about an order of magnitude greater than the desired $P_{fa}$. The automatic settings of the specifications is carried out by sampling the SNR-$\cos^2 \theta$ plane of the desired $d(\beta)$ mesa-plot over four uniformly spaced values of the mismatch, namely $\lambda_i \in \{1, 0.75, 0.5, 0.25\}$, and over four different values of SNR $\gamma_i \in \{8, 10, 15, 20\}$, for a total of $S = 16$ specifications.  

\begin{figure}
\centering
\includegraphics[width=0.4\textwidth]{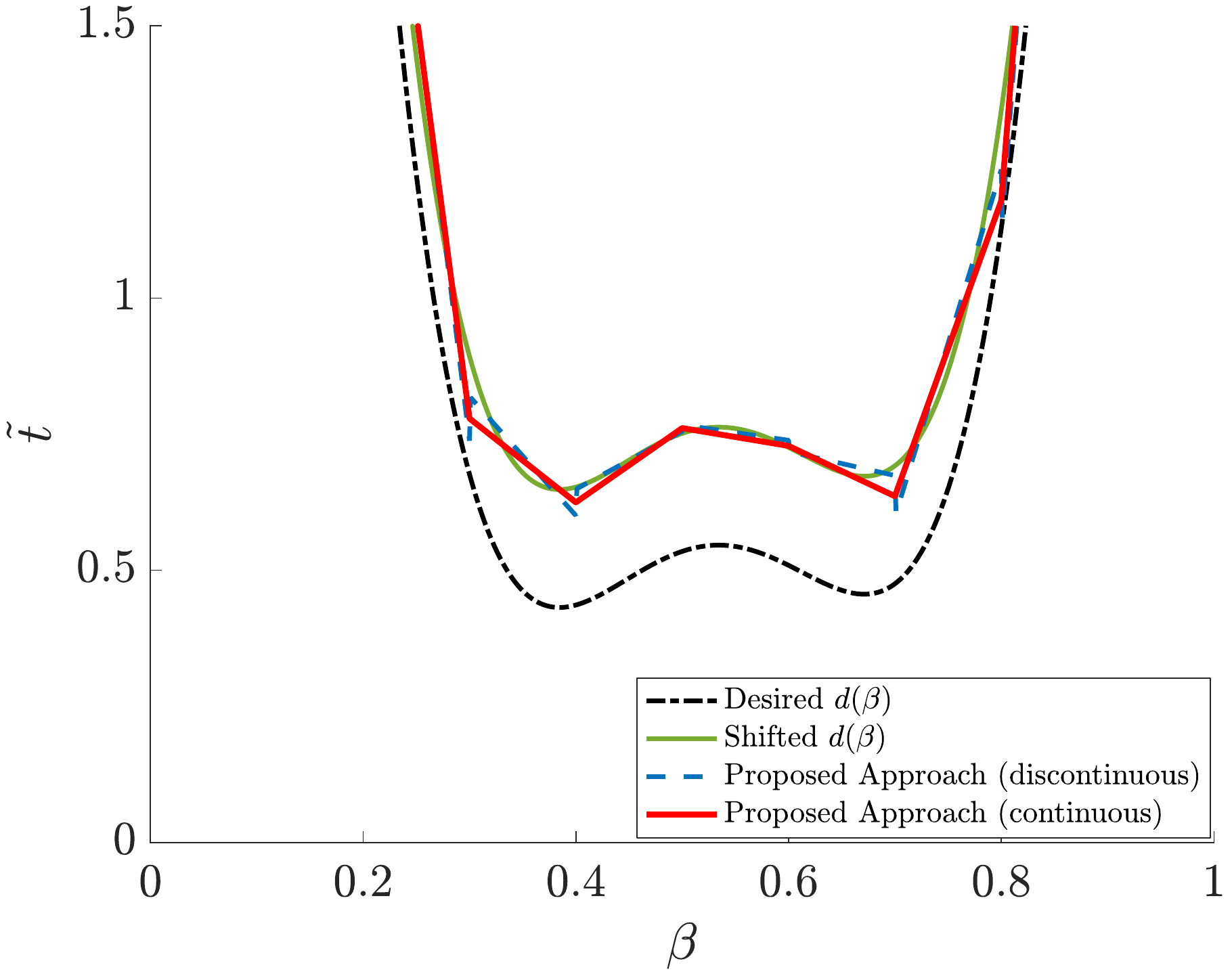}
\caption{CFAR-FP for the case of the double-well detector.  }\label{fig:CFAR-FP_doublewell}
\end{figure}

\begin{figure}
\centering
\includegraphics[width=0.33\textwidth]{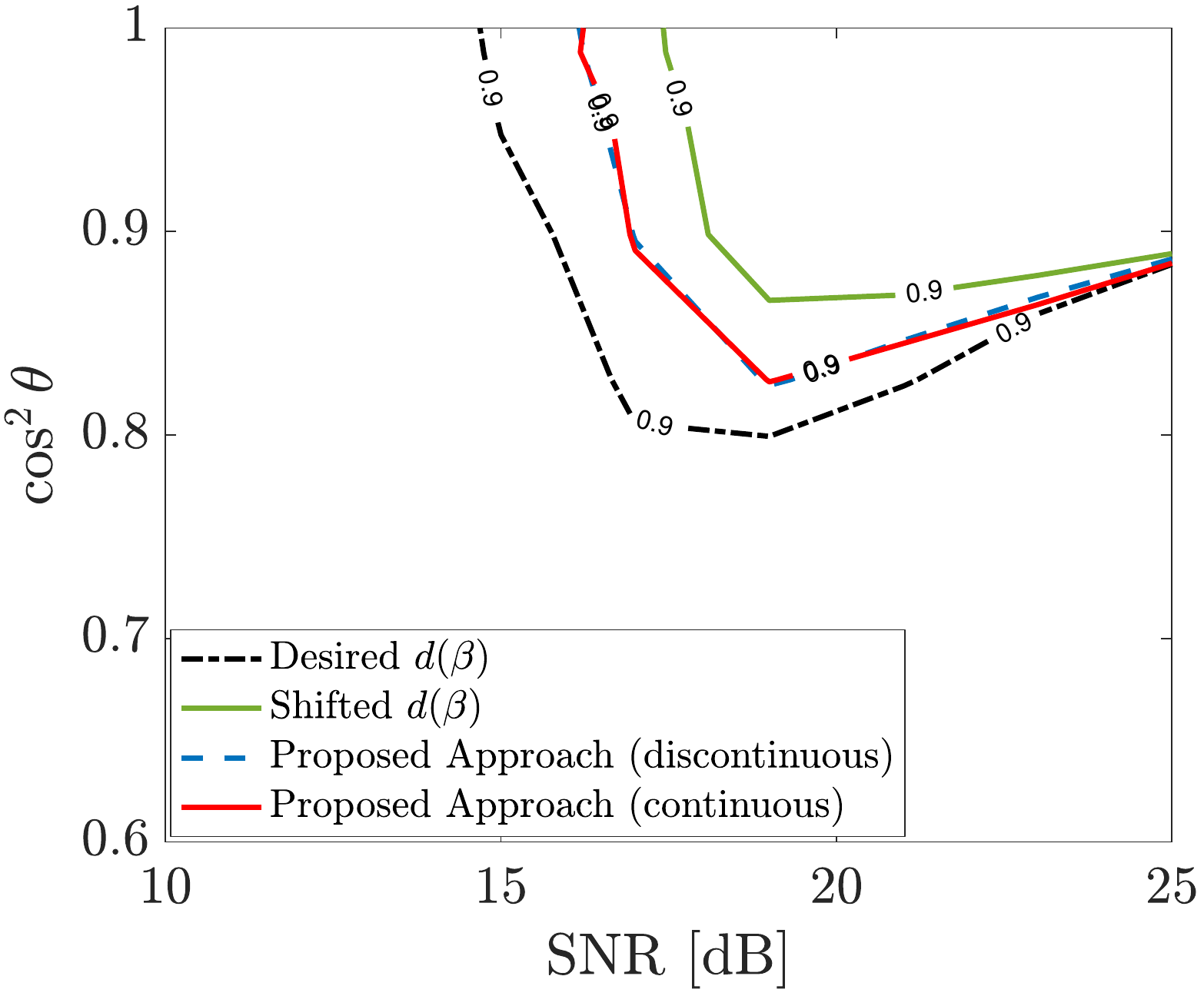}
\caption{Detection performance of the desired $d(\beta)$ in comparison to the shifted $d(\beta)$ and to the proposed approach (with either continuous or discontinuous decision region boundary) at $P_d = 0.9$.  }\label{fig:Pd09_double}
\end{figure}

\subsubsection{\ed Analysis of the decision region boundary approximation}
In Fig.~\ref{fig:CFAR-FP_doublewell}, we depict the decision region boundary $f_{\bm{m}}(\beta; \bm{\epsilon}^*)$ of the double-well detector obtained through the proposed approach, which for this case returned a $f_{\bm{m}}(\beta; \bm{\epsilon}^*)$ with $k^* = 10$ segments. For comparison, the  decision  region  boundary  of the detector obtained by simply shifting $d(\beta)$ to match the preassigned $P_{fa}$ is reported, labeled as ``Shifted $d(\beta)$'' for brevity. As it can be noticed, the proposed approach provides a decision region boundary which is close to that of the shifted $d(\beta)$, for both the discontinuous and continuous versions. 
Indeed, the discontinuity gaps at the junctions of the different segments are very limited, so if one is interested in having a continuous boundary the solution $f_{\bm{m}}(\beta; \bm{\epsilon}^*)$ provided by Algorithm~\ref{alg} can be safely made continuous through the refinement stage described in Algorithm~\ref{alg_continuous}, without compromising the performance of the design.
The latter  remarkable fact is not apparent from Fig.~\ref{fig:CFAR-FP_doublewell} but can be easily appreciated by looking at the mesa-plots of the different detectors, as discussed below.

\subsubsection{\ed Comparison between continuous and discontinuous solutions}
In Fig.~\ref{fig:Pd09_double}, we report the detection performance of the proposed double-well detector in the SNR-$\cos^2 \theta$ plane for a single level of $P_d = 0.9$, also in comparison with the desired $d(\beta)$ and with the shifted $d(\beta)$. 
It can be seen that the proposed double-well detector follows more closely the behavior of the desired $d(\beta)$, with an evident gap compared to the performance of the shifted $d(\beta)$. This confirms the validity of the proposed algorithm: by exploring a range of $k$-segment piecewise-linear approximations of the desired $d(\beta)$ from a coarse ($k =2$) up to a fine scale ($k = p=16$), our design approach is able to provide a satisfactory trade-off between satisfaction of the $P_{fa}$ constraint and deviation from the desired performance in terms of $P_d$ under matched and mismatched conditions.
From Fig.~\ref{fig:Pd09_double} it also emerges that the minimal changes required to make the decision region boundary $f_{\bm{m}}(\beta; \bm{\epsilon}^*)$ continuous practically lead to zero deviations from the detection performance of the discontinuous case; therefore, in the following we will no longer consider such a distinction.
The competitor ``shifted $d(\beta)$'' provides instead a much worse approximation of the desired $d(\beta)$; this may appear counterintuitive, since in Fig.~\ref{fig:CFAR-FP_doublewell} the decision region boundaries of the three detectors look quite close to each other. But this is exactly the motivation of this work, as discussed in Sec.~\ref{sec:background}: given the highly-nonlinear mapping between the maximal invariant statistics (CFAR-FP) and the corresponding detection performance under matched and mismatched conditions (mesa-plots), it is challenging to adjust the desired $d(\beta)$ so that the $P_{fa}$ constraint can be fulfilled without jeopardizing the design, as far as possible.

\subsubsection{\ed Analysis of the deviation from the desired behavior}
It is worth remarking that the desired $d(\beta)$ do not share the same $P_{fa}$ of the other detectors ($d(\beta)$ is working at a $P_{fa}$ which is about an order of magnitude greater), hence it cannot be strictly considered as a benchmark for the desired $P_d$, but only for its behavior, which should be approximated as closely as possible. For a better assessment, we consider as metric the area of the planar region delimited by the iso-$P_d$ curves of a given  detector and that of the $d(\beta)$ in the SNR-$\cos^2 \theta$ plane, which measures the level of vicinity between such curves. In fact, the smaller the gap from the iso-$P_d$ curves of $d(\beta)$, the better the approximation of the desired behavior under both matched and mismatched conditions. We will refer to such a metric as ``area between iso-$P_d$ curves'' (AbI). 
\begin{figure}
\includegraphics[width=0.48\textwidth]{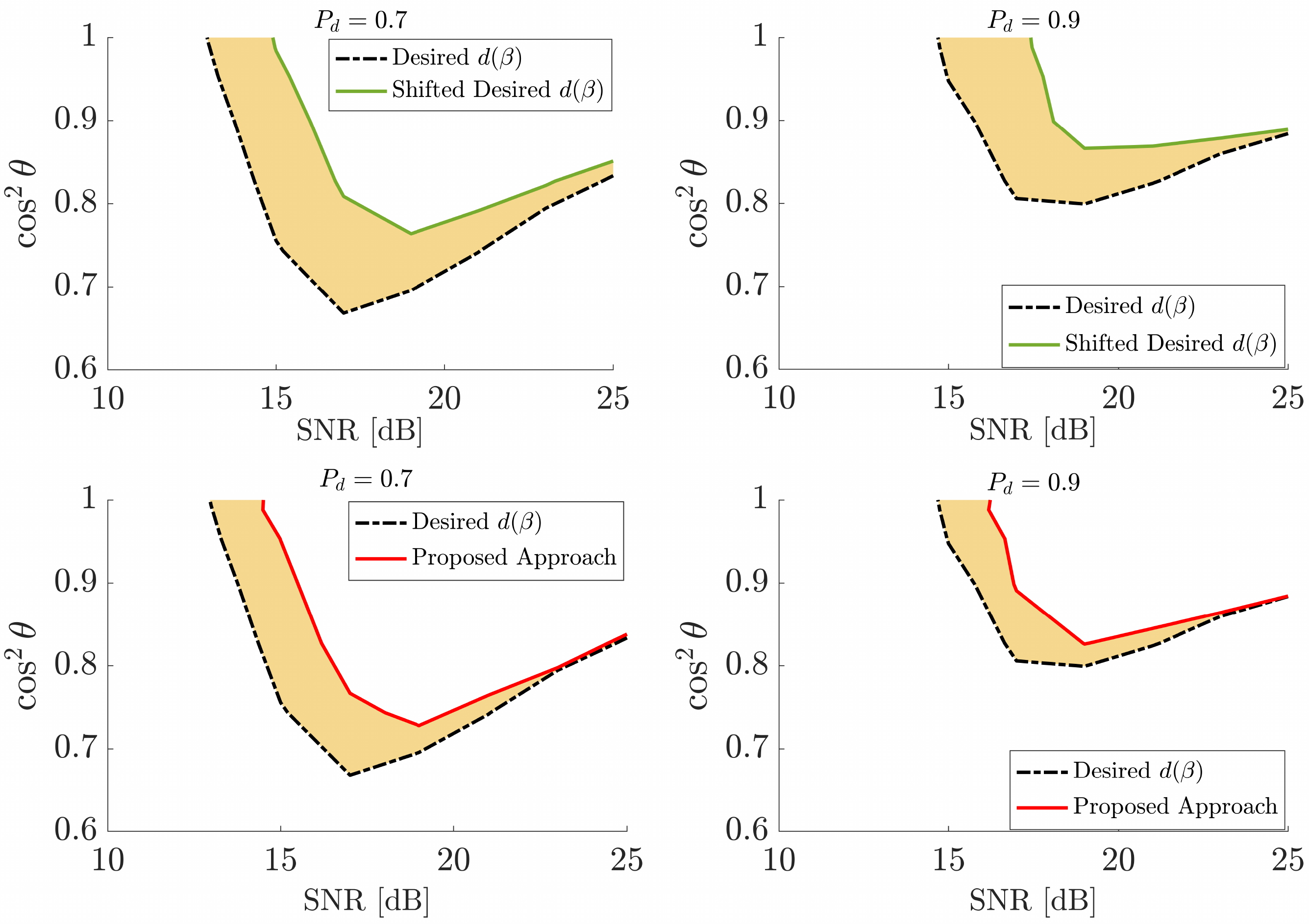}
\caption{Areas between iso-$P_d$ curves (AbI) for the proposed approach, in comparison with the shifted $d(\beta)$ detector.  }\label{fig:Comp_Pd_areas_double}
\end{figure}
\begin{figure}
\includegraphics[width=0.45\textwidth]{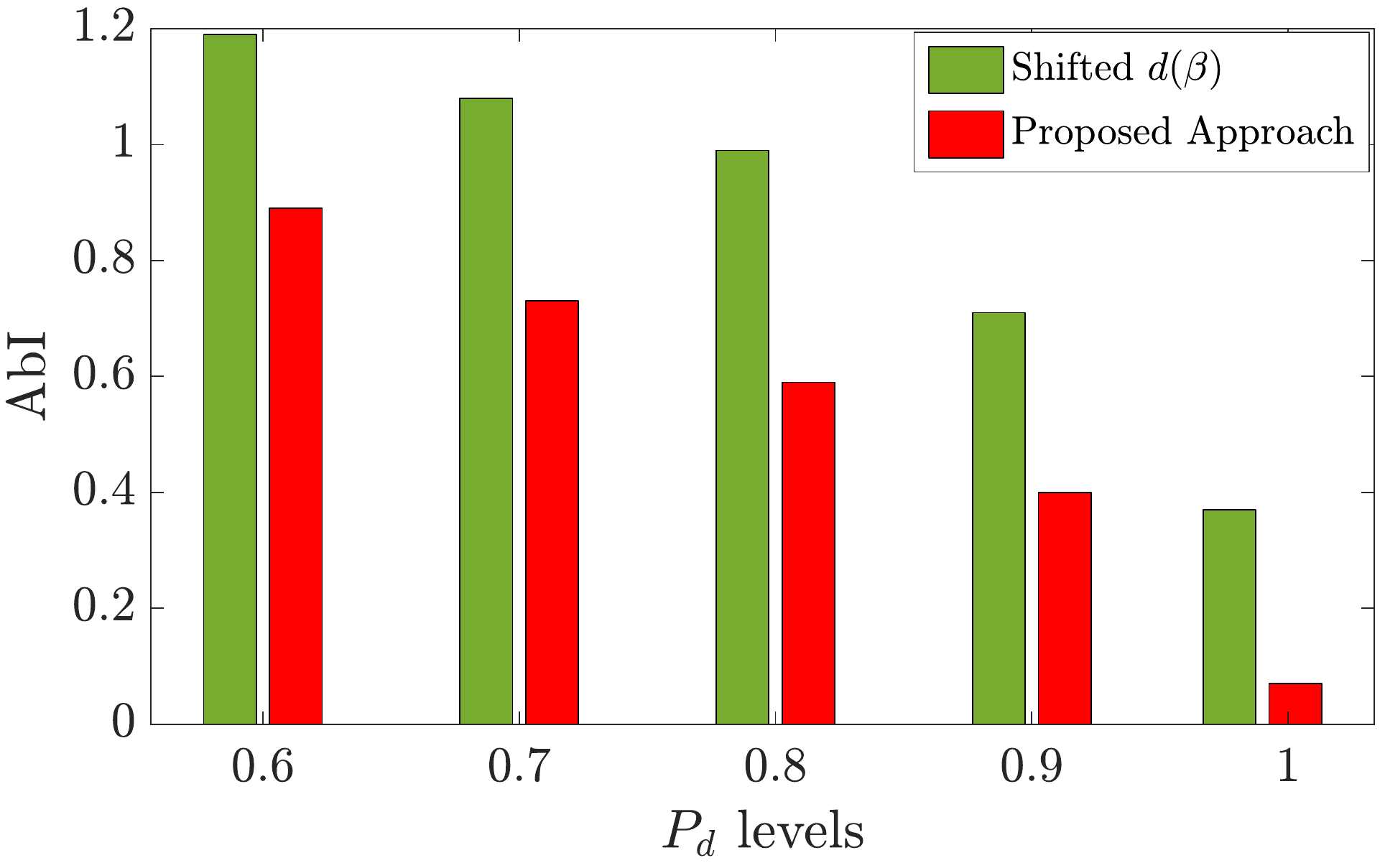}
\caption{Areas between iso-$P_d$ curves (AbI) of the double-well detectors as a function of the $P_d$ levels.  }\label{fig:Bars_areas_double}
\end{figure}

\begin{figure*}
\centering
\includegraphics[width=0.60\textwidth]{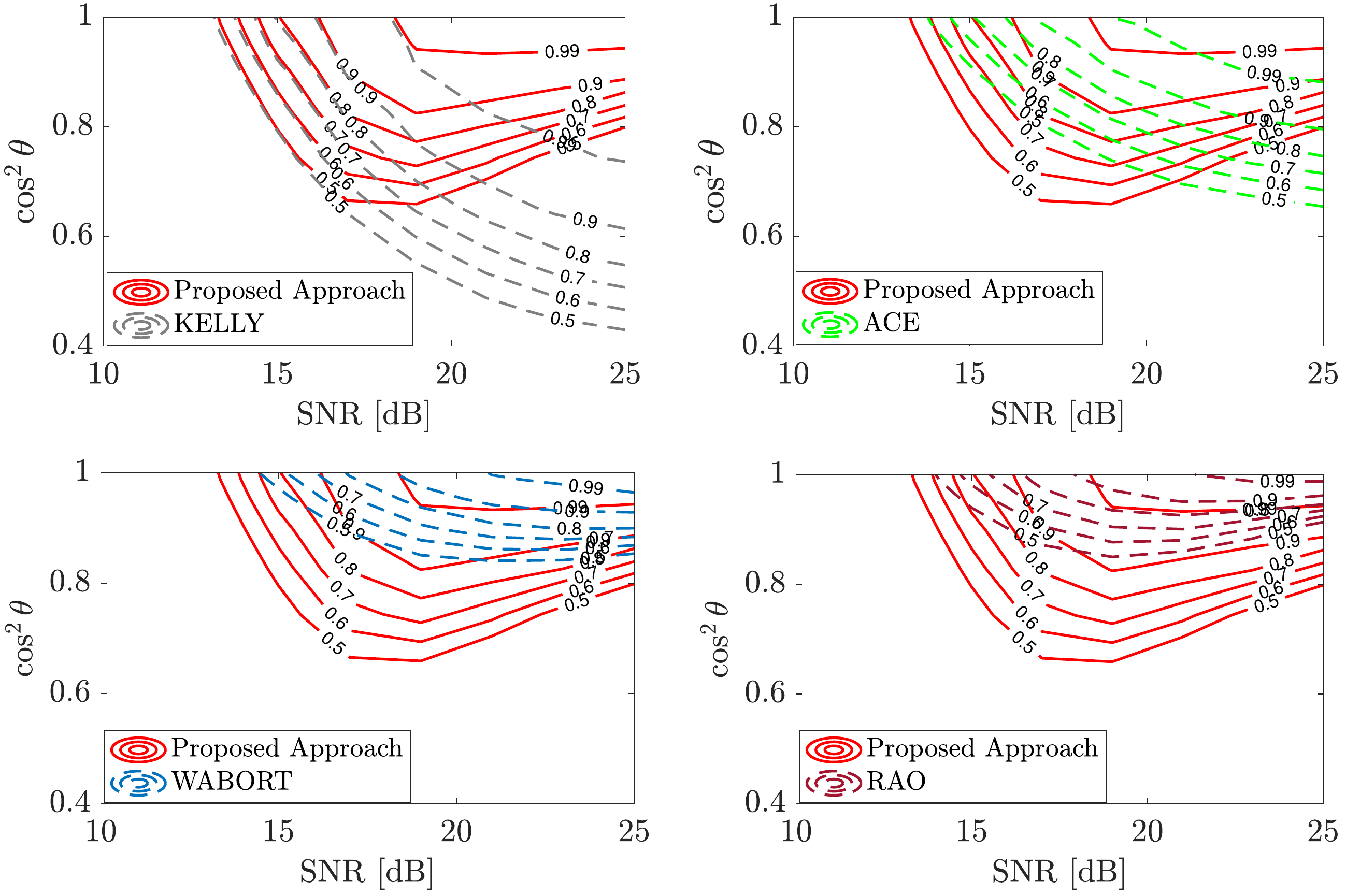}
\caption{Mesa plots of the double-well detector compared to Kelly, ACE, WABORT and RAO detectors.}\label{fig:Mesa_double}
\end{figure*}
In Fig.~\ref{fig:Comp_Pd_areas_double} we highlight the AbI of the proposed detector in comparison with the AbI of the shifted $d(\beta)$, for two different levels of $P_d = \{ 0.7, 0.9\}$.  Results
demonstrate that the double-well detector provides a better approximation of the desired $d(\beta)$, being its corresponding areas visibly smaller than those of the shifted $d(\beta)$ detector. More precisely,  Fig.~\ref{fig:Bars_areas_double} shows the exact values of the AbI as a function of the $P_d$ levels. Remarkably, the  proposed  double-well  detector  satisfactorily  follows  the   detection performance of the desired $d(\beta)$ for all the considered $P_d$ levels, with an approximation error that tends to decrease for higher values of the $P_d$.

\subsubsection{\ed Analysis of the detection performance}
In Fig.~\ref{fig:Mesa_double}, we compare the double-well detector against state-of-the-art detectors. Since the former customized detector aims at rejecting mismatched signals while preserving high detection power under matched conditions,
we have included as relevant competitors the well-known Kelly's detector, which is taken as a reference for the performance under matched conditions, as well as the ACE, WABORT and RAO detectors, which are instead taken as benchmarks for
the performance under mismatched conditions. It is interesting to observe that the double-well detector guarantees almost the same $P_d$ of Kelly’s detector under matched conditions, while providing
much higher levels of selectivity. Furthermore, it is much more powerful than ACE and exhibits
higher rejection capabilities for large SNR values. It is also remarkable to notice that, compared to the very selective WABORT and RAO detectors, the double-well detector does not experience any significant $P_d$ loss under matched conditions, while still preserving satisfactory rejection capabilities.

\subsection{Performance of the Combined AMF-ROB Detector }\label{perf::AMFROB}

\begin{figure}
\includegraphics[width=0.4\textwidth]{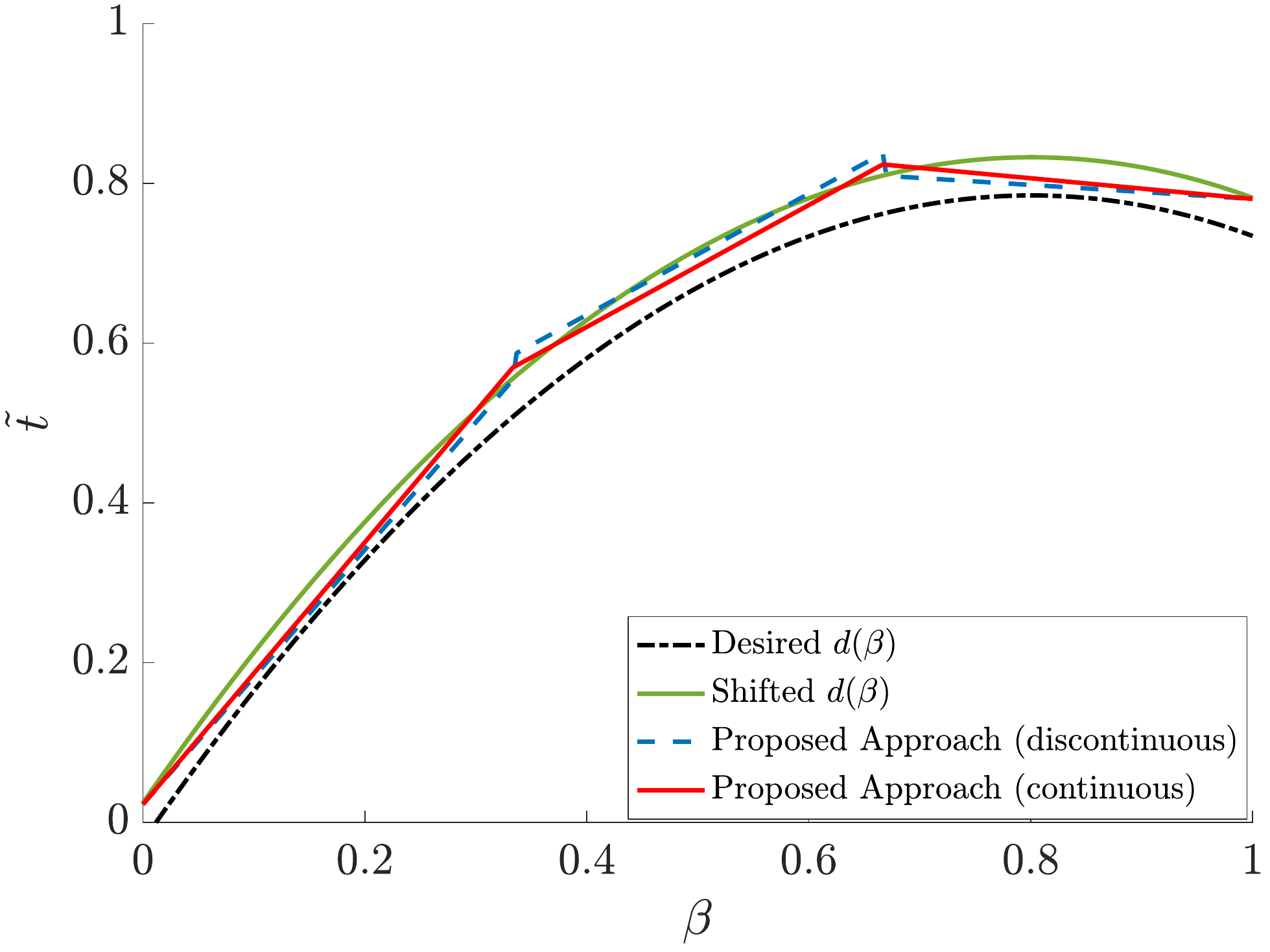}
\caption{CFAR-FP for the case of the combined AMF-ROB detector. }\label{fig:CFAR-FP_combined}
\end{figure}

For further illustration, we analyze also the second example of detector design presented in Sec.~\ref{sec::2}, namely the case in which the desired $d(\beta)$ is obtained by combining the decision region boundaries of the two well-known detectors AMF and ROB. {\ed For the analysis, we consider the same parameters as in Sec.~\ref{sec::double_well}.} Intuitively, this scenario appears to be more favorable since the desired $d(\beta)$ exhibits a $P_{fa} = 2.6 \cdot 10^{-4}$, which is already quite close to the final desired $P_{fa}$.  The specifications $\bm{\rho}_i$s are automatically set by sampling the SNR-$\cos^2 \theta$ plane of the desired $d(\beta)$ mesa plots at the same coordinates of the previous example in Sec.~\ref{sec::double_well}.

\begin{figure}
\includegraphics[width=0.43\textwidth]{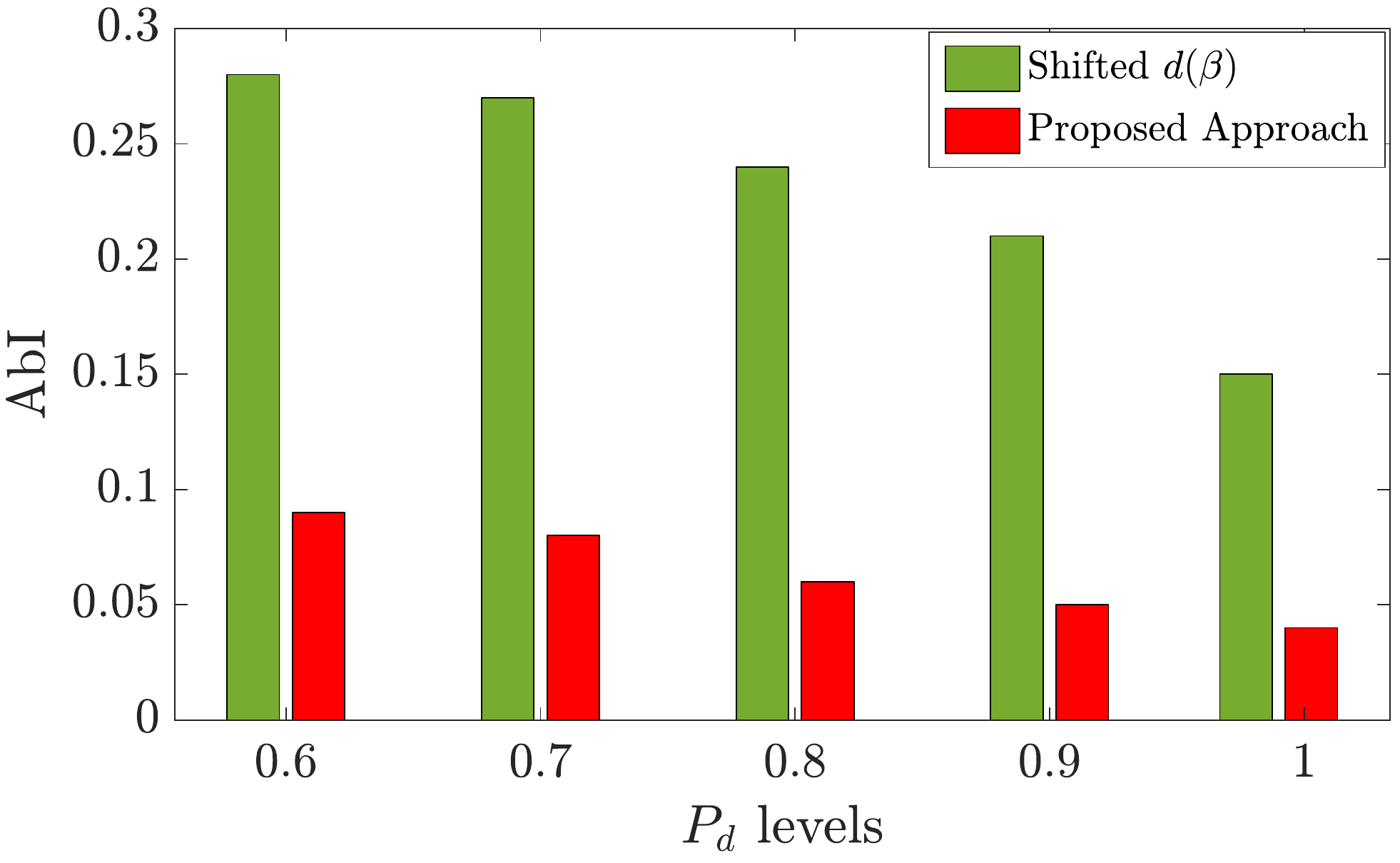}
\caption{Areas between iso-$P_d$ curves (AbI) of the combined AMF-ROB detectors as a function of the $P_d$ levels.  }\label{fig:Bars_areas_combined}
\end{figure}
\subsubsection{\ed Analysis of the decision region boundary approximation}
Fig.~\ref{fig:CFAR-FP_combined} depicts the decision region boundary of the combined detector obtained through the proposed reduced-complexity algorithm, in comparison with the desired $d(\beta)$ and with the decision region boundary of the shifted $d(\beta)$. Among  the explored configurations of $k \in [2,16]$, the best value of the objective function is achieved for $k^* = 3$, namely the configuration in which the decision region boundary $f_{\bm{m}}(\beta; \bm{\epsilon}^*)$ consists of a juxtaposition of three linear segments approximating the positive-slope line of AMF for the lower range of $\beta$ and attaining ROB's behavior in the upper range. 
\subsubsection{\ed Comparison between continuous and discontinuous solutions} Also in this case, the discontinuities in $f_{\bm{m}}(\beta; \bm{\epsilon}^*)$ are very limited, as confirmed by the continuous version of the decision region boundary, which practically coincides with its  discontinuous version.

\begin{figure*}
\centering
\includegraphics[width=0.60\textwidth]{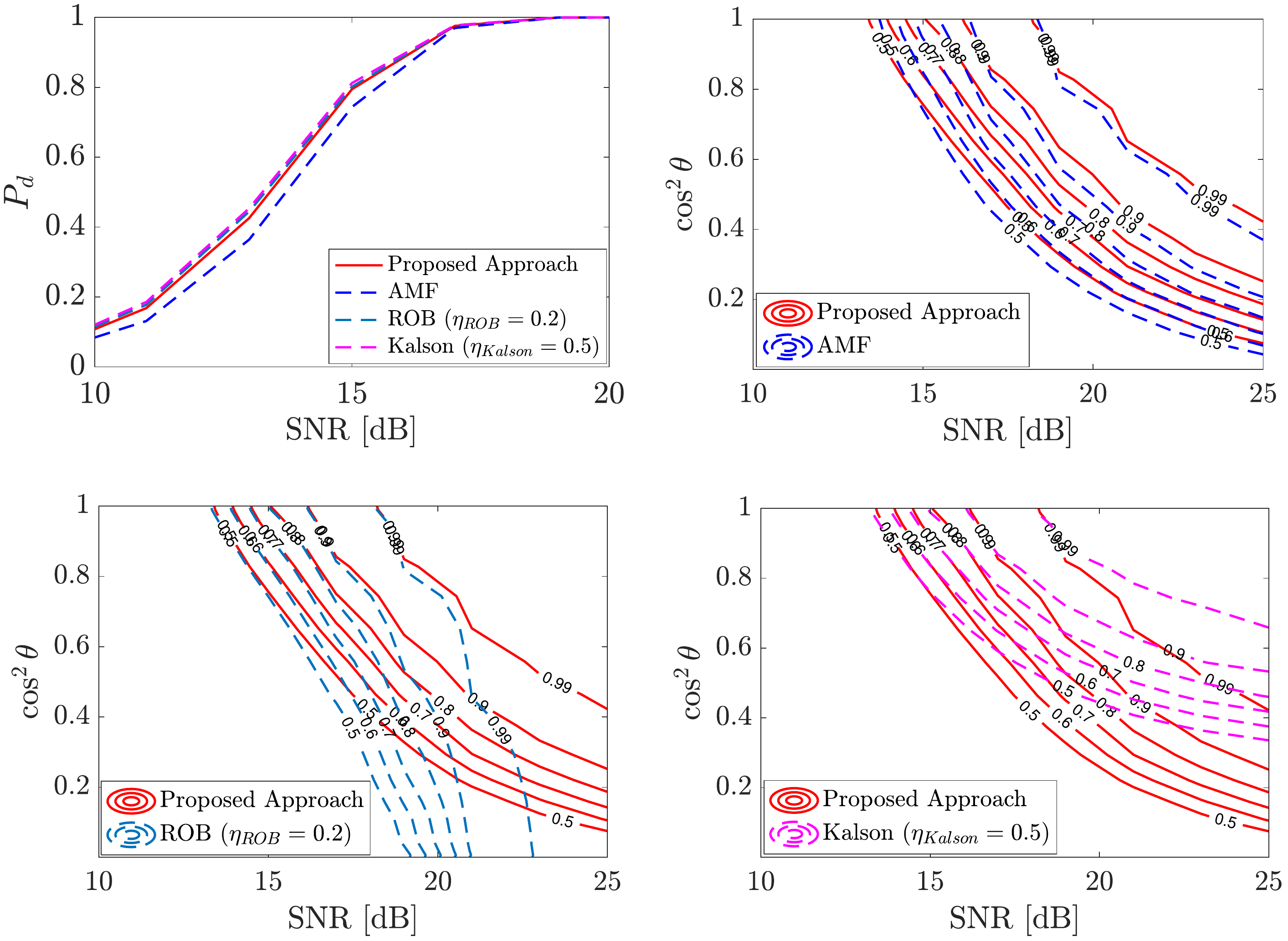}
\caption{Mesa plots of the combined AMF-ROB detector compared to the AMF, ROB, and Kalson detectors. }\label{fig:comp_robusti}
\end{figure*}
Notice that, as for the previous case of the double-well detector, the  decision region boundaries of the different detectors are close to each other, and in this case also closer to the desired $d(\beta)$ since the $P_{fa}$ of the latter is already near the chosen design value. However, again, a small difference in the boundary can produce a non-negligible impact on the detection performance.
\subsubsection{\ed Analysis of the deviation from the desired behavior} Specifically, to assess the adherence of the combined detector to the behavior of the desired $d(\beta)$, in Fig.~\ref{fig:Bars_areas_combined} we report the AbI values as a function of the $P_d$ levels. Results demonstrate that the proposed approach is able to provide a more accurate approximation of the desired $d(\beta)$ compared to the shifted $d(\beta)$, even in this more favorable scenario in which the $P_{fa}$ of the desired $d(\beta)$ is very close to the preassigned one. As a whole, we can conclude that the use of the proposed methodology can yield closer approximations of a desired detector $d(\beta)$ for a prefixed $P_{fa}$ compared to the mere shift of $d(\beta)$ itself, regardless of the extent of the gap between initial and desired $P_{fa}$ values.

\subsubsection{\ed Analysis of the detection performance} For completeness, we finally report in Fig.~\ref{fig:comp_robusti} the performance of the proposed combined detector, in comparison to the performance of the AMF, ROB, and Kalson detectors. It is worth noting that the proposed  detector is able to combine the satisfactory robustness of the AMF with the high detection power of the ROB under matched conditions (which is practically the same as Kelly's detector). Specifically, it is as powerful as both ROB and Kalson under matched conditions, while guaranteeing at the same time the level of robustness of AMF under mismatched conditions, which is  in between ROB and Kalson.

\subsection{\ed Evaluation on Real Data}

{\ed
To corroborate the above results, we have carried out a performance evaluation on real radar measurements by considering the L-band land clutter data collected by the Phase One radar at the Katahdin Hill site, MIT Lincoln Laboratory. We used the dataset contained in the \emph{``H067038.3"} file, which is composed of 30720 temporal returns from 76 range cells with HH polarization \cite{PhaseOne1,PhaseOne2}. Given that the total number of real clutter measurements is not sufficient to guarantee a number of snapshots matching the $100/P_{fa}$ rule for the value of $P_{fa} = 10^{-4}$ assumed in the simulation analyses, we  readjusted the design of the two proposed detectors by downscaling it to $N = 4$, $K = 8$, and a desired $P_{fa} = 10^{-3}$. All the remaining parameters are instead kept the same. 

For both the double-well and combined AMF-ROB CFAR detectors, the decision region boundaries obtained through the proposed low-complexity design procedure
show a very good match with the desired ones (figures omitted for conciseness), confirming the general validity of our approach, which can flexibly adapt to a new set of parameters. In the following we report the analysis of the corresponding performance.

\begin{figure}
		\centering
		\includegraphics[width=\columnwidth]{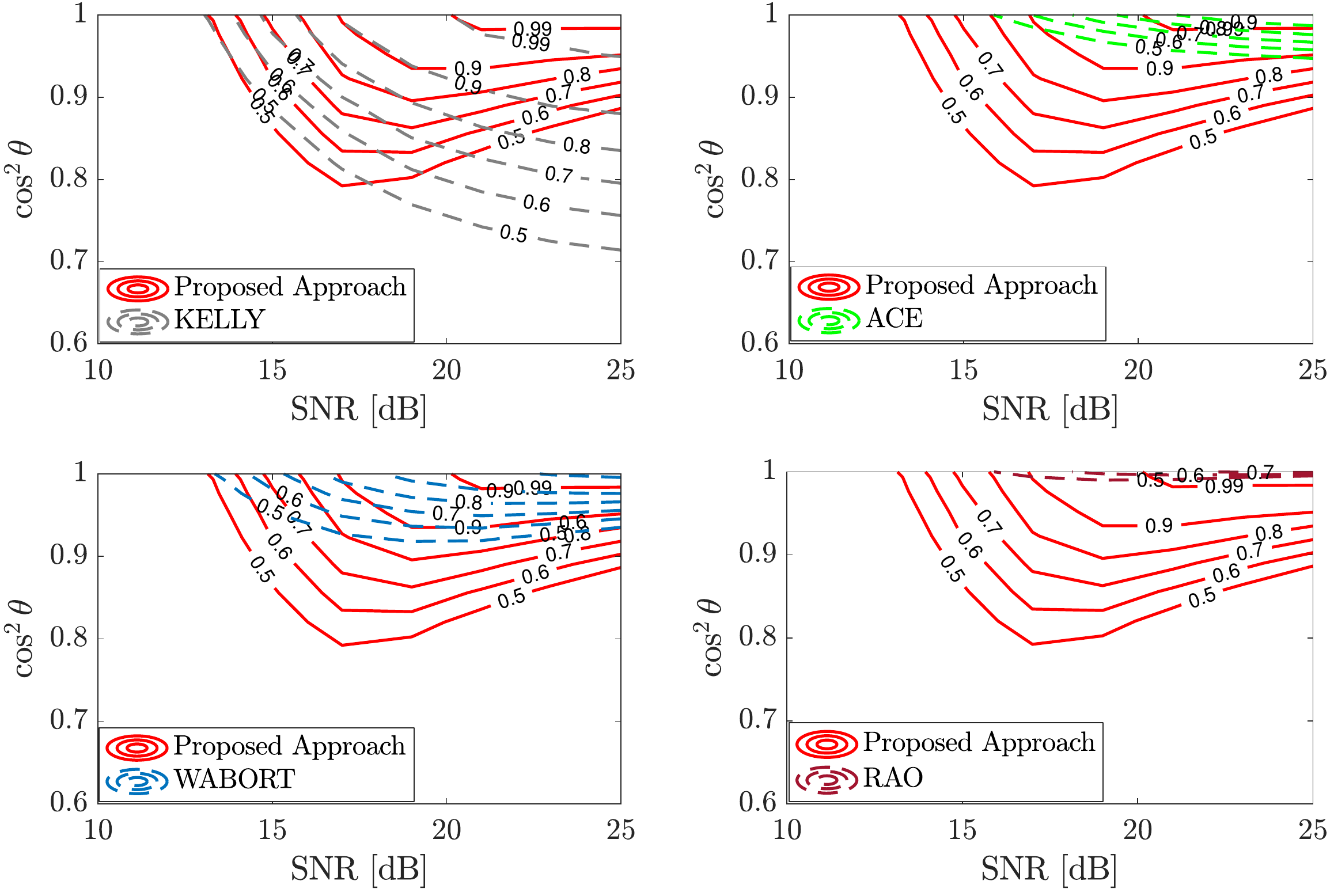}
		{\ed \caption{Mesa plots of the double-well detector compared to Kelly, ACE, WABORT and RAO detectors for $N=4$, $K=8$, evaluated on the Phase One dataset.}
		\label{fig:Mesa_doublewell_real}}
	\end{figure}
	\begin{figure}
		\centering
		\includegraphics[width=\columnwidth]{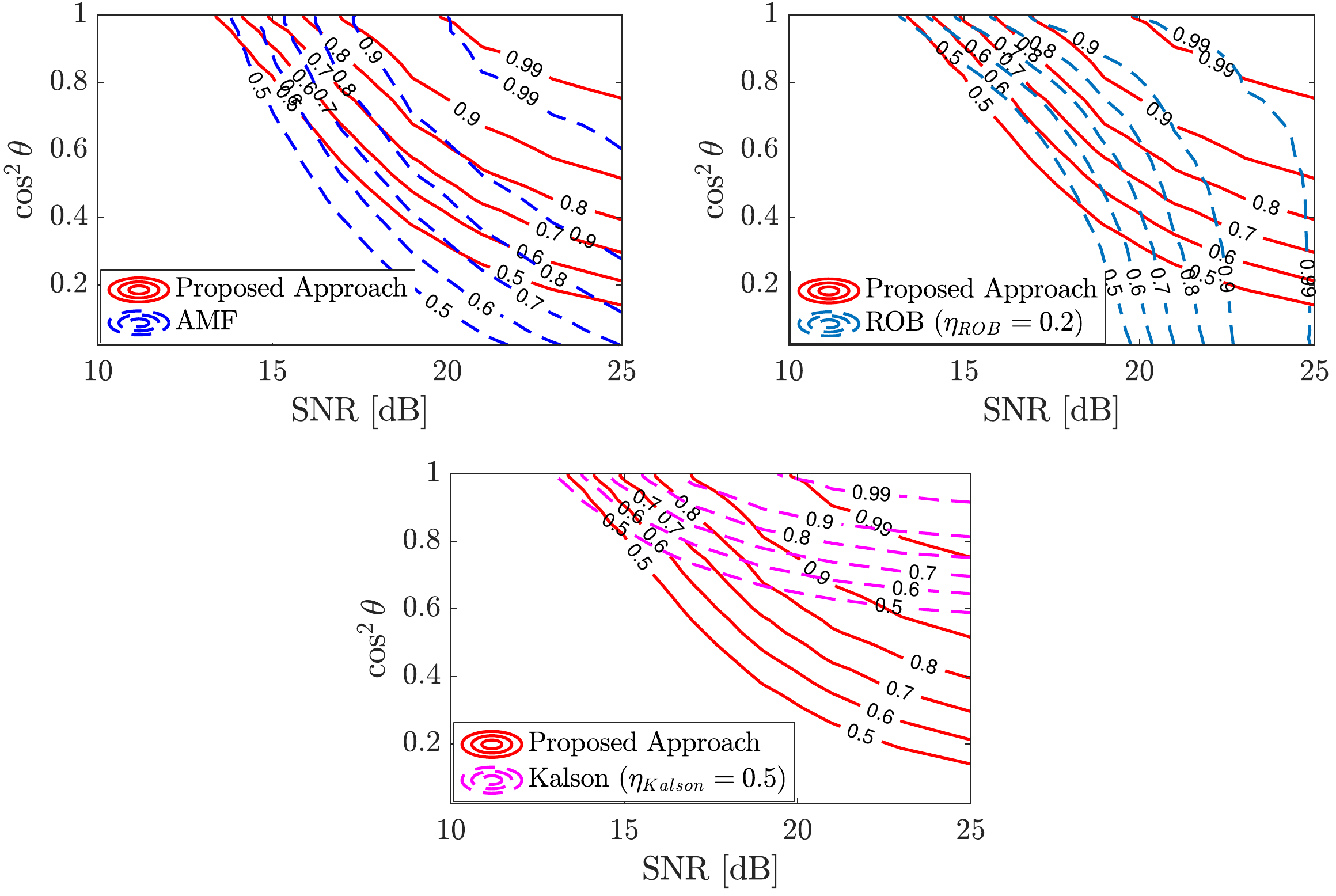}
		{\ed \caption{Mesa plots of the combined AMF-ROB detector compared to the AMF, ROB, and Kalson detectors for $N=4$, $K=8$, evaluated on the  Phase One dataset.}
		\label{fig:Mesa_AMFROB_real}}
	\end{figure}
	
We start the evaluation by estimating the actual $P_{fa}$ of the detectors when operating on the Phase One dataset. All the detectors thresholds are set to guarantee the desired $P_{fa} = 10^{-3}$ on simulated data. For the analysis, we select the 30-th range cell as the CUT and the $K/2$ adjacent range cells on each side of the CUT as secondary data. We construct the data vectors by selecting $N$ consecutive pulses from each range cell, with 1 pulse of overlap in order to obtain $10^4$ different snapshots, and perform the detection. The proposed double-well detector exhibits a $P_{fa} = 5 \cdot 10^{-4}$, a value that does not deviate too much from the nominal one and, remarkably, is even lower. Kelly and ACE detectors share the same value of $P_{fa} = 5 \cdot 10^{-4}$, while WABORT and RAO exhibit a $P_{fa} = 9 \cdot 10^{-4}$. Interestingly, all the selective detectors share very similar values of $P_{fa}$, which well approximate the nominal one. Analogous results have been obtained for the case of robust detectors: specifically, the combined AMF-ROB, AMF, and ROB detectors also exhibit a $P_{fa} = 5 \cdot 10^{-4}$, while Kalson has a $P_{fa} = 4 \cdot 10^{-4}$.

The performances in terms of $P_d$ are then evaluated by adding to the CUT a synthetic target $\alpha \bm{v}$, as done for the simulated data. To make the comparison precisely fair, we adjusted the thresholds of the WABORT, RAO, and Kalson so as to match the same $P_{fa} = 5 \cdot 10^{-4}$ of the other detectors. In Fig.~\ref{fig:Mesa_doublewell_real} and Fig.~\ref{fig:Mesa_AMFROB_real},  we compare the performance of the proposed double-well and combined AMF-ROB detectors against the same state-of-the-art competitors considered in Sec.~\ref{sec::double_well} and Sec.~\ref{perf::AMFROB}. As it can be noticed, the proposed design procedure confirms its effectiveness in correctly approximating the desired CFAR behaviors, with the proposed double-well detector that preserves the high $P_d$ of the Kelly's detector under matched conditions, while striking a more evident selectivity under mismatched conditions. Remarkably, the gain in terms of $P_d$ under matched conditions becomes much more pronounced compared to the ACE and RAO detectors, which turn out to be very selective. Similarly, the proposed combined AMF-ROB detector practically keeps the satisfactory robustness of the AMF detector, while guaranteeing the high detection power of the ROB under matched conditions. Its robust behavior still lies between ROB and Kalson, confirming the same findings on simulated data.}

\section{Conclusion}\label{sec:conclusions}

The paper has proposed a methodology for the optimal design of customized CFAR detectors in Gaussian disturbance, having desired behavior in terms of robustness or selectivity, and working at a preassigned $P_{fa}$ value. 
By exploiting a reinterpretation of CFAR detection in a suitable feature plane based on maximal invariant statistics (CFAR-FP), the optimal approximation problem has been formulated and analyzed. To overcome its analytical and numerical intractability, a general reduced-complexity approach has been developed, which seeks for a solution in a restricted, but sufficiently rich feasible set. Analytical expressions have been provided to ease the computational aspects, by assuming a piecewise-linear structure for the approximated decision region boundary.
The proposed algorithm  is very efficient in finding a detector satisfying the requirements and provides satisfactory performance, though it explores only a limited part of the solution space, which is otherwise extremely large.

Under this framework, two novel customized detectors have also been  designed and analyzed, for illustration purposes. 
One of them combines two existing detectors, and is able to achieve an intermediate robust behavior by taking the best from both, as desired. {\ed Indeed, since Kelly's is recognized as the ``best'' detector in terms of $P_d$ under matched conditions (close to the locally most powerful invariant detector \cite{Bose_Steinhardt}), one should try to combine the characteristics of Kelly's decision boundary with that of a more robust or selective detector, according to the desired behavior. Being ROB practically as powerful as Kelly's detector under matched conditions but very robust, we used a combination of AMF (which has a lower robustness) and ROB to find a good balance. The same approach can work for the selective case, by appropriate choice of the detectors to combine. In the second provided design example we followed a different strategy, that is to conceive the ``shape'' of the desired detector  based on considerations on the spread of point clusters under $H_0$ and $H_1$ for both matched and mismatched conditions, showing how a completely ad-hoc detector can be designed through the proposed methodology. In particular, an ad-hoc curve is drawn in the CFAR-FP whose shape allows for rejection of mismatched signals, but at the same time leads to high detection power under matched conditions.

Future work includes the investigation of alternative resolution approaches for problem \eqref{eq::inf_dim} that can consider additional feasible solutions, as this may lead closer to obtaining the ``best'' detector for a given set of specifications. Another interesting direction is to investigate possible extensions of the proposed framework to environments different from the homogeneous Gaussian.  
Actually, since most detectors of interest can be written as a function of a set of proper maximal invariant statistics \cite{DeMaio2015,Ciuonzo2016,Tang2020}, the rationale used in this paper can be extended by considering feature spaces of arbitrary dimensions: in fact, although for more than three dimensions a direct graphical representation is not possible, still in principle the curve-approximation problem turns into a hypersurface-approximation problem for which hyperplane-based tiles can be adopted instead of the piecewise-linear template to generalize the ``greedy'' logic of the proposed design algorithm (i.e., to adjust $P_{fa}$ by modifying one piece of the decision boundary at a time). More research is however needed to characterize clusters behaviors in such different feature spaces and identify suitable expressions to compute $P_{fa}$ (and, possibly, also  $P_d$).}

\appendices

\section{Weighting Strategies}\label{app:discussion}

As to the weights $w_i\in [0,1]$, different strategies can be identified to end up with an objective function able to encode the different design requirements (and their relative priority) in the proper way, summarized as follows:
\begin{itemize}
\item $w_i = 1, \forall i$
  \item $w_i=w(\gamma_i,\lambda_i)$
    \item $w_i=w(\bm{\epsilon}; \bm{\rho}_1,\ldots,\bm{\rho}_S), \forall i$.
\end{itemize}
The first option considers constant weights, hence is tantamount to having an (unweighted) least squares functional, which will treat all the $S$ specifications $\bm{\rho}_i$ equally. As a consequence, errors on higher probabilities of detection will have a dominant effect, being their impact on the objective function greater in magnitude.

The second option considers diversified weights given by a weighting function $w(\cdot)$ depending only on SNR and match/mismatch level specifications; its definition impacts on the relative priority given to the different specifications. The choice $w_i = 1/\psi_i^2$ has the special meaning of a relative error, thus would overcome the limitation of the first option ($w_i=1\, \forall i$). However, it would produce the opposite effect of giving too much weight to errors on low values of the probability of detection (under matched or mismatched conditions), which may be detrimental to the detection power for higher SNRs (an unacceptable behavior for a radar detector). So, other choices should be identified, e.g. an increasing function of $\lambda_i$, possibly constant in (a certain range of)  $\gamma_i$.
In fact, $w_i = \lambda_i$ (or $\lambda_i^\alpha$, $\alpha >0$, to also adjust the decay rate) would promote more priority to maximizing $P_d$ for any value of SNR, but at the same time would take into account the desired probability of deciding for $H_1$ under various level of mismatches. 
For $S=1$ with $\lambda_1=1$, $\psi_1=1$, and a chosen SNR $\gamma_1$, the minimization of $\mathcal{C}_1(\bm{\epsilon})$ would be  equivalent to the maximization of
\begin{align*}
P_{d} &= P(\tilde{t} >  f_{\bm{m}}(\beta; \bm{\epsilon})| \gamma_1,1) \nonumber \\
&= 1 - \int_0^1 F_{\tilde{t}|\beta, H_1} (  f_{\bm{m}}(\beta; \bm{\epsilon}) |\gamma_1) p(\beta) d\beta
\end{align*}
where $F_{\tilde{t}|\beta, H_1}(\cdot)$ is the CDF of $\tilde{t}$ given $\beta$ under the $H_1$ hypothesis,
hence the Neyman-Pearson inspired rationale is retrieved. A slight generalization of this formulation is the maximization of $P_d$ in a span of SNRs, with weights all equal or increasing/decreasing according to the given priority to the high/low SNR regime. 

In general, the main drawback of the second weighting strategy ($w_i=w(\gamma_i,\lambda_i)$) is that a fine-tuning of the weights is necessary, whose impact may be not completely predictable, thus resulting in a trial-and-error effort. For this reason, in the following we adopt the third strategy, in which weights are all equals but set to a certain function $w(\bm{\epsilon}; \bm{\rho}_1,\ldots,\bm{\rho}_S)$ that depends on all specifications as well as on the optimization parameters $\bm{\epsilon}$. A simplifying yet reasonable choice is to map such dependencies into the error function $e(\cdot)$ as
\begin{equation}
w_i =  \frac{1}{S}  \underbrace{\sqrt{\frac{1}{S-1}\sum_{i=1}^S \left(e(\bm{\epsilon}; \bm{\rho}_i) - \bar{e}(\bm{\epsilon}; \bm{\rho}_1,\ldots,\bm{\rho}_S)\right)^2}}_{\sigma_e}
\end{equation}
where $\bar{e}(\bm{\epsilon}; \bm{\rho}_1,\ldots,\bm{\rho}_S) = 1/S \sum_{j=1}^S e(\bm{\epsilon}; \bm{\rho}_i)$ and $\sigma_{e}$ denotes the empirical standard deviation of the squared error. By substituting this expression into \eqref{eq:P1}, we finally obtain \eqref{eq::finalcost}.

\section{Proof of Proposition \ref{Prop1}}\label{app:A}

We start from $\Psi(f_{m_i}(\beta; \epsilon_i))$, which can be  obtained as a more convenient rewriting of  $F_{\tilde{t}|\beta}(f_{m_i}(\beta; \epsilon_i))$. More precisely, we recall that in the general case \cite{BOR-MorganClaypool}
\begin{align*}
F_{\tilde{t}|\beta}&(f_{m_i}(\beta; \epsilon_i)) = \frac{f_{m_i}(\beta; \epsilon_i)}{(1+f_{m_i}(\beta; \epsilon_i))^{K\! -\! N\! +\! 1}}\sum_{k=0}^{K-N}\! \binom{K\!- \! N\! +\! 1}{1\! +\! k} \nonumber \\
& \times (f_{m_i}(\beta; \epsilon_i))^k \e^{-\frac{\delta^2_F}{1+f_{m_i}(\beta; \epsilon_i)}}\sum_{i=0}^k \left( \frac{\delta^2_F}{1 + f_{m_i}(\beta; \epsilon_i)}\right)^i \frac{1}{i!}.
\end{align*}
We then notice that, by manipulating  from \cite[eq. 6.5.13]{abramowitz} the innermost summation can be expressed in terms of Eulerian complete and upper incomplete gamma functions as
$$
\sum_{i=0}^k \left( \frac{\delta^2_F}{1 + f_{m_i}(\beta; \epsilon_i)}\right)^i \frac{1}{i!} = \e^{\frac{\delta^2_F}{1+f_{m_i}(\beta; \epsilon_i)}} \frac{\Gamma(1+k,\frac{\delta^2_F}{1+f_{m_i}(\beta; \epsilon_i)})}{\Gamma(1+k)}.
$$
The ultimate $\Psi(f_{m_i}(\beta; \epsilon_i))$  follows by plugging back the above  expression in $F_{\tilde{t}|\beta}(f_{m_i}(\beta; \epsilon_i))$ and by performing a change of variable $1+k = \ell$.

Similarly, $\Omega(\beta)$ can be obtained by a proper rewriting of $p(\beta)$. Specifically, we recall that
\begin{align*}
p(\beta) = &\frac{\e^{-\delta^2_\beta \beta}\beta^{K-N+1}}{(1-\beta)^{2-N}}\frac{\Gamma(K+1)}{\Gamma(K-N+2)}\sum_{j=0}^{K-N+2}\binom{K-N+2}{j} \nonumber \\
&\times \frac{1}{(N+j-2)!}\delta^{2j}_\beta (1-\beta)^{j}.
\end{align*}
The expression of $\Omega(\beta)$ can be derived by exploiting the similarity between the above summation and the generalized Laguerre polynomials $L^{(\alpha)}_n(x)$ of order $n=K-N+2$ for $\alpha=N-2$; in particular, by using \cite[eq. 13.6.9]{abramowitz} together with \cite[eq. 22.3.9]{abramowitz} it is possible to 
write the following identity:
\begin{align*}
\sum_{j=0}^{K-N+2} & (-1)^{j} \binom{K}{K-N+2-j} \frac{x^j}{j!} \\ 
&= \binom{K}{N-2} {}_1F_1(-K+N-2,N-1; x).
\end{align*}
Then, noticing that binomial/factorial terms in the two expressions can be related as
$$
\frac{\binom{K-N+2}{j}}{(N+j-2)!}  =  \binom{K}{K-N+2-j} \frac{(K-N+2)!}{K! j!}
$$
and that $(1-\beta)^j = (-1)^j (\beta-1)^j$,  we  obtain for $x=\delta_\beta^2(\beta-1)$
\begin{align*}
\sum_{j=0}^{K-N+2} &
\frac{\binom{K-N+2}{j}}{(N+j-2)!}  \delta^{2j}_\beta (1-\beta)^{j} \\
& = \frac{1}{\Gamma(N-1)}   {}_1F_1(-K+N-2,N-1;\delta_\beta^2(\beta-1)) 
\end{align*}
since $\frac{(K-N+2)!}{K! } \binom{K}{N-2} = \frac{1}{(N-2)!} = \frac{1}{\Gamma(N-1)}$.
The thesis follows by substituting back into $p(\beta)$ and recognizing that $\frac{\Gamma(K+1)}{\Gamma(K-N+2)\Gamma(N-1)} = \binom{K}{N-2}$.

\section{Proof of Proposition \ref{Prop2}}\label{app:B}
Under the $H_0$ hypothesis, we have that
\begin{align*}
r_i(\epsilon_i) = \int_{(i-1)/p}^{i/p}& \frac{f_{m_i}(\beta; \epsilon_i)}{(1+f_{m_i}(\beta; \epsilon_i))^{K-N+1}}\sum_{j=0}^{K-N}\binom{K-N+1}{1+j} \nonumber \\
&\times (f_{m_i}(\beta; \epsilon_i))^j p(\beta)  d\beta.
\end{align*}
By using a change of variable $1+j = z$ and noting that by the algebraic binomial formula $\sum_{z=0}^{K-N+1} \binom{K-N+1}{z}(f_{m_i}(\beta; \epsilon_i))^z = (1+f_{m_i}(\beta; \epsilon_i))^{K-N+1}$, $r_i(\epsilon_i)$ can be recast as
\begin{align*}
r_i(\epsilon_i) &= \int_{(i-1)/p}^{i/p} p(\beta)d\beta\\
&- \left( \int_{0}^{i/p} (1+f_{m_i}(\beta; \epsilon_i))^{-(K-N+1)}p(\beta) d\beta \right. \nonumber \\
& -\left. \int_{0}^{(i-1)/p} (1+f_{m_i}(\beta; \epsilon_i))^{-(K-N+1)}p(\beta) d\beta \right).
\end{align*}
Considering the integrals between braces, for $u=i/p$ or $u=(i-1)/p$ we have
\begin{align*}
&\int_{0}^{u} (1+f_{m_i}(\beta; \epsilon_i))^{-(K-N+1)}p(\beta) d\beta = \frac{K!}{(N\! -\! 2)! (K\! -\! N\! +1)!} \\
&\times \int_0^u (1+\epsilon_i+m_i\beta)^{-(K-N+1)} \beta^{K-N+1} (1-\beta)^{N-2} \mathrm{d}\beta
\end{align*}
which is a generalization of the integral definition of the Euler's Beta function; in particular, it can be computed by exploiting the following identity 
$$
\int_0^u \frac{x^n(1-x)^m}{(1+ax)^n}  \mathrm{d}x = \frac{x^{n+1}}{n+1} F_1(n+1,-m,n, n+2; u, -au)
$$
which is valid for $a>0$ (for $a=0$ returns the Beta function) and $0\leq u\leq 1$, $m$ and $n$ positive integers,
where $F_1(a,b,c,d; y,z)$ is the Appell $F_1$ (hypergeometric) function of two variables \cite[sec. 9.18]{grad}.

The result of Proposition \ref{Prop2} follows by noting that $\int_{(i-1)/p}^{i/p} p(\beta)d\beta = F_{\beta|H_0}(i/p) - F_{\beta|H_0} ((i-1)/p)$ and exploiting the identity above for $a=m_i/(1+\epsilon_i)$, $m=N-2$, and $n=K-N+1$, i.e., 
\begin{align*}
&\int_{0}^{u} (1+f_{m_i}(\beta; \epsilon_i))^{-(K-N+1)}p(\beta) d\beta \\
&= \frac{u}{K-N+2} \left(\frac{u}{1+\epsilon_i}\right)^{K-N+1} \nonumber \\
& \times F_1 \!\Big( \!K\! -\! N\! +2,2\! - \! N,K\! - \! N\! +1,K\! - \! N\! +3; u, -\frac{m_i u}{1+\epsilon_i}\Big)
\end{align*}
hence the thesis follows straight.

\section*{Acknowledgement}

The authors wish to thank Prof. J. B. Billingsley, MIT Lincoln Lab, for the Phase One L-band clutter data.

\end{document}